\documentclass[journal]{IEEEtran}

\usepackage{amsthm}
\usepackage{amssymb}
\usepackage{amsfonts}
\usepackage{amsfonts,subfigure,multicol,color,verbatim,graphicx,cite,epsfig,amssymb,amsmath,cases,bm,xcolor,multirow,array}
\usepackage{algorithm,algorithmicx,algpseudocode}
\usepackage{multirow}
\usepackage{multicol}
\usepackage{supertabular}
\usepackage{array}
\usepackage{colortbl}
\usepackage{longtable}
\usepackage{subfigure}
\usepackage{dsfont}
\usepackage{url}

\usepackage{epstopdf,booktabs}

\newtheorem{lemma}{Lemma}

\allowdisplaybreaks[4]

\usepackage{soul}
\usepackage{color, xcolor}
\soulregister{\cite}7
\soulregister{\citep}7
\soulregister{\citet}7
\soulregister{\ref}7
\soulregister{\pageref}7
\soulregister{\eqref}7

\UseRawInputEncoding
\begin{document}

\title{Beam Management in Low Earth Orbit Satellite Communication With Handover Frequency Control and Satellite-Terrestrial Spectrum Sharing }

\author{Yaohua~Sun,~Jianfeng~Zhu,~and~Mugen~Peng,~\IEEEmembership{Fellow,~IEEE}
	\thanks{Jianfeng Zhu (jianfeng@bupt.edu.cn), Yaohua Sun (sunyaohua@bupt.edu.cn), and Mugen Peng (pmg@bupt.edu.cn) are with the State Key Laboratory of Networking and Switching Technology, Beijing University of Posts and Telecommunications, Beijing 100876, China. (\bf{Corresponding author: Mugen Peng})}}
\maketitle

\begin{abstract}

To achieve ubiquitous wireless connectivity, low earth orbit (LEO) satellite networks have drawn much attention.
However, effective beam management is challenging due to time-varying cell load, high dynamic network topology, and complex interference situations.
In this paper, under inter-satellite handover frequency and satellite-terrestrial/inter-beam interference constraints, we formulate a practical beam management problem, aiming to maximize the long-term service satisfaction of cells.
Particularly, Lyapunov framework is leveraged to equivalently transform the primal problem into multiple single epoch optimization problems, where virtual queue stability constraints replace inter-satellite handover frequency constraints.
Since each single epoch problem is NP-hard, we further decompose it into three subproblems, including inter-satellite handover decision, beam hopping design and satellite-terrestrial spectrum sharing.
First, a proactive inter-satellite handover mechanism is developed to balance handover frequency and satellite loads.
Subsequently, a beam hopping design algorithm is presented based on conflict graphs to achieve interference mitigation among beams, and then a flexible satellite-terrestrial spectrum sharing algorithm is designed to satisfy the demands of beam cells and improve spectral efficiency.
Simulation results show that our proposal significantly improves service satisfaction compared with baselines, where the average data queue length of beam cells is reduced by over 50$\%$ with affordable handover frequency.

\end{abstract}


\section{Introduction}

In last decade, rapid technical developments have facilitated the deployments and applications of LEO satellite networks~\cite{background2,background3,background4}.
Although LEO satellite networks can provide seamless service, effective beam management is highly challenging.
Specifically, due to uneven traffic distributions, limited radio resource, and high dynamic topology, LEO satellite beam management need to deal with frequent inter-satellite handover\cite{handover_twc}, unsatisfied capacity by beam hopping\cite{multi_sat2}, and severe inter-beam interference.
Meanwhile, considering the shortage of available spectrum resource, satellite-terrestrial spectrum sharing is a promising direction\cite{space_spectrum_sharing1}.

\subsection{Related works}
\subsubsection{Inter-satellite handover}
Low handover frequency can reduce signaling overhead and provide continuous service experience.
In \cite{handover_rssi1}, users monitor the received signal strength, and handover procedure is performed when the measured signal strength is lower than a certain threshold.
However, it is intractable to judge appropriate handover occasions because near-far distance effects are insignificant in LEO satellite networks\cite{handover_twc}.
Therefore, several other handover criteria are proposed, including longest remaining service time~\cite{handover_graph}, shortest distance~\cite{min_distance} and load balance~\cite{{handover_load1}}.
In contrast to adopting single criterion with partial information for handover, Zhang \emph{et al}. in \cite{multi-attribute} develop a multi-attribute handover decision scheme based on entropy by combining multiple handover criteria.
Although the above works can reduce handover frequency, handover decisions are mainly made by users with local observations about LEO satellite networks, which may result in handover to a satellite with high load.
To overcome this issue, a centralized handover strategy is proposed in \cite{handover_centralized} by searching the minimum cost and maximum flow in a binary graph between satellites and users.
Wang \emph{et al}. in \cite{handover_twc} propose a conditional handover {triggering} mechanism to enhance service continuity, where gateway stations reallocate serving satellites for all users when network topology changes.
{The simulation results show that a centralized handover strategy with proactive handover triggering mechanism can improve service continuity.
\textbf{Unfortunately, since current inter-satellite handover mechanisms are only triggered when network topology changes, they are shortsighted and cannot adapt to time-varying traffic arrival, which may lead to load imbalance among satellites and degrade network throughput.}}

\subsubsection{Beam hopping}
With inevitable challenges caused by uneven transmission demand distribution and potential inter-beam interference, industry and academia have adopted beam hopping as a promising approach to fully utilize time-frequency resource and support inter-beam spectrum sharing.
Overall, beam hopping research can be divided into two categories, including intra-satellite beam hopping strategies \cite{one_sat2,one_sat3} and inter-satellite cooperative beam hopping strategies\cite{multi_sat2,multi_sat3}.
Authors in \cite{one_sat2} propose a swap matching algorithm to reduce the gap between traffic requests and capacity provision.
In \cite{one_sat3}, beam cells are dynamically divided into multiple clusters with comparable loads, and then a beam hopping strategy is developed under inter-cluster interference isolation.
Compared to scenarios with a single satellite, inter-satellite beam hopping design not only requires intra-satellite interference mitigation but also needs to consider inter-satellite interference.
To this end, a multi-satellite beam hopping design approach is developed in \cite{multi_sat2} by intra-satellite beam hopping and additional inter-satellite interference coordinating, where serving relationships between satellites and beam cells are determined by a load balance method.
In \cite{multi_sat3}, an enhanced heuristic method for beam hopping is developed in dual satellite cooperative scenarios, which contributes to improving spectral efficiency.
\textbf{Although current beam hopping research can obtain a high service satisfaction, their beam pattern design mainly depends on instantaneous transmission demand without considering potential inter-satellite handover and interference to terrestrial networks.
}

\subsubsection{Satellite-terrestrial spectrum sharing}
Due to limited spectrum resource, LEO satellite network may be deficient to offer sufficient capacity to hop-spot areas via its owner spectrum.
Since terrestrial base stations sometimes operate in low-load states, spectrum sharing between LEO satellite and terrestrial networks is an attractive solution to tackle the spectrum shortage problem.
A protection area based method is adopted in \cite{space_spectrum_sharing2} to achieve spatial isolation between satellites and terrestrial networks.
Moreover, authors in \cite{space_spectrum_sharing1} evaluate satellite-terrestrial interference and then calculate the radius of protected area for the coexistence of satellite and terrestrial networks.
Unfortunately, interference avoidance methods in \cite{space_spectrum_sharing1,space_spectrum_sharing2} cause coverage holes and disrupt communication services in protected area.
For seamless coverage, Du \emph{et al}. in \cite{frequency_spectrum_sharing2} adopt a flexible spectrum sharing strategy based on a second-priced auction mechanism between a satellite and several base stations, achieving high data rate.
In \cite{noma_spectrum_sharing1}, a spectrum sharing method is designed by jointly allocating spectrum and power resource for satellite-terrestrial networks.
\textbf{Although the above works provide feasible solutions for satellite-terrestrial spectrum sharing and improve network capacity by satellite-terrestrial interference mitigation, they mainly focus on scenarios with a single satellite or a single beam, whereas base stations may receive multiple interference signals from different beams or satellites in practice.
Moreover, with the increasing number of base stations and satellites, their proposed methods cannot be directly applied due to sizeable computational complexity.
}

\subsection{Motivation and contribution }

Overall, the long-term service satisfaction of cells is highly relevant to LEO satellite beam management approaches, which needs to consider inter-satellite handover, beam hopping design, and satellite-terrestrial spectrum sharing.
As mentioned,  due to the time-varying transmission demand of cells, current inter-satellite handover mechanisms are ineffective in tackling the challenges caused by imbalanced load among satellites accumulated over time.
To satisfy the differential transmission demands of cells, the above studies in beam hopping design greedily maximize the service satisfaction of cells in each epoch.
However, they do not take inter-satellite handover and interference to terrestrial networks into account.
In addition, most satellite-terrestrial spectrum sharing studies only focus on a simple scenario with a single satellite or a single beam, which can not be applied in dynamic multiple LEO satellite networks.
\textbf{To the best of our knowledge, long-term beam management approach, which jointly considers inter-satellite handover decision, beam hopping design, and satellite-terrestrial spectrum sharing, has not been studied.}

To fill this gap, this paper proposes an effective long-term beam management approach for scenarios with multiple LEO satellites, comprising inter-satellite handover frequency control, beam hopping design, and satellite-terrestrial spectrum sharing.
To handle the challenge caused by time-averaged objective and handover frequency constraints, we leverage Lyapunov framework to obtain beam management decisions by solving a sequence of single epoch problems.
In each epoch, we identify inter-satellite handover events by the proposed conditional handover triggering mechanism, which keeps load balance among satellites and maintains a low inter-satellite handover frequency.
Under the given serving satellites, we further develop low-complexity beam hopping design and satellite-terrestrial spectrum sharing algorithms to maximize service satisfaction.
The main contributions of this paper are summarized as follows.
\begin{itemize}
    \item
    \textbf{Practical network model with multiple LEO satellites and novel beam management problem formulation:}
    We study a practical network model with multiple LEO satellites, which feature dynamic topology, random traffic arrival, earth-fixed cells, satellite-terrestrial/inter-beam interference, and corresponding interference mitigation to achieve inter-beam and satellite-terrestrial spectrum sharing.
    Based on such a model, we further formulate a novel beam management problem to maximize the long-term service satisfaction of beam cells, where inter-satellite handover frequency, interference constraints, and the guaranteed transmission demand of terrestrial networks are considered.
     Considering the difficulty to deal with the time-averaged objective function and constraints directly, we construct virtual queues for beam cells and then adopt Lyapunov framework to transform the primal problem into a series of tractable single epoch problems, where queue stability constraints replace inter-satellite handover frequency constraints.
    Since the transformed problem has NP complexity, we further decompose it into three simplified subproblems, consisting of inter-satellite handover decision, beam hopping design, and satellite-terrestrial spectrum sharing.

    \item
    \textbf{Effective algorithm design for the above three subproblems:}
    First, a novel conditional handover triggering mechanism is proposed, which establishes the connections between inter-satellite handover events and network statuses, including time-frequency resource utilization, cell service satisfaction, and network topology.
    With such a mechanism, we develop an inter-satellite handover decision algorithm based on a multi-attribute handover decision scheme and swap matching, aiming to control long-term inter-satellite handover frequency and satellite load distributions.
    Subsequently, we build a conflict graph to characterize potential inter-beam interference under the given serving satellite of each beam cell.
    Then, a low-complexity beam hopping design algorithm is developed to maximize the transmission capacity of cells and satisfy interference constraints.
    After beam hopping design, we propose a satellite-terrestrial spectrum sharing algorithm by combining improved binary sparrow search and greedy search, which contributes to high service satisfaction and spectral efficiency.

    \item
    \textbf{Deep analysis of computational complexity and extensive simulation evaluations:}
    The complexities of the proposed three algorithms and complete beam management approach are analyzed in-depth, which indicates that our proposal reduces computational complexity from an exponential level to a square level.
    By conducting massive simulations, it is shown that the proposed beam management approach can satisfy maximum inter-satellite handover frequency constraint.
    Furthermore, our proposal reduces the average data queue length by over 84.37$\%$ compared with traditional inter-satellite handover mechanisms.
    In addition, the average data queue length is reduced by over 50$\%$ compared with the baselines of satellite-terrestrial spectrum sharing.
    Moreover, the maximum network capacity can reach 98.51$\%$ of theoretical upper-bound performance without interference, which verifies the superiority of our proposal.

\end{itemize}

The reminder of this paper is organized as follows.
Section~\ref{sec:system_model} presents the system model and spectrum sharing schemes for LEO satellite networks.
In Section~\ref{sec:problem formulation}, the concerned beam management problem is formulated, and the proposed beam management algorithms are presented in Section~\ref{sec:Algorithm Framework}.
Simulation results are given in Section~\ref{sec:sim}, and the conclusions are drawn in Section~\ref{sec:con}.

\section{SYSTEM MODEL}\label{sec:system_model}

\subsection{Network model}\label{sec:scenario}

 \begin{figure}[t]
    \centering
    \includegraphics[width=3.3in]{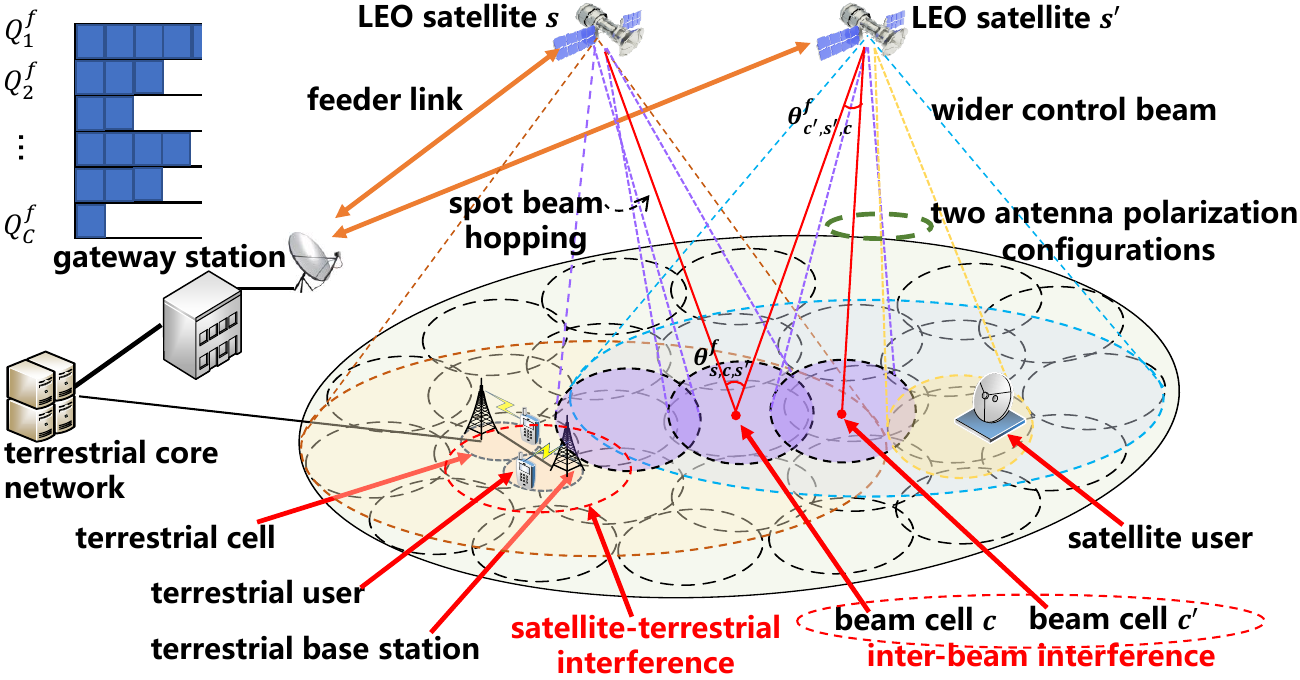}
    \caption{The concerned LEO satellite network scenario.}\label{fig1}
\end{figure}

As show in Fig.~\ref{fig1}, define $\mathcal{S}=\{1,...,s...,S\}$ as the set of LEO satellites, and each satellite can generate up to $B$ spot beams to provide high data rates.
Define $\mathcal{B}=\{1,..,b,...,B\}$ as the beam index set of a satellite.
Satellite antennas adopt two polarization configurations to alleviate inter-beam interference \cite{cross_polar_isolation}.
Moreover, all spot beam directions are controlled by a gateway station, which locates on the ground and connects satellites via feeder links.
Due to the long propagation delay, the gateway station formulates beam management decisions in advance, which consists of the serving satellites of beam cells, spot beam hopping design, and the bandwidth of beams.
Meanwhile, several wider control beams are generated by LEO satellites to broadcast beam management decision to satellite users.
Control beams and spot beams operate on orthogonal frequency bands, and hence interference between them can be ignored.
Denote the set of earth-fixed beam cells and terrestrial cells as $\mathcal{C}=\{1,2,...,C\}$ and $\mathcal{G}=\{1,2,...,G\}$, respectively.
Denote the bandwidth of satellite and terrestrial networks as $W_1$ and $W_2$, respectively.
Terrestrial cells adopt soft frequency reuse schemes to serve handheld terminals, and satellite users support carrier aggregation between the spectrums of LEO satellite networks and terrestrial networks.
Gateway station can exchange information with the terrestrial core network through optical fiber links, including cell loads and locations as well as beam hopping design.

Define $f=1,2,3,...$ as the index of beam scheduling epoch.
Each epoch is further divided into $T$ identical slots with duration $	T_{slot}$~\cite{multi_sat2}, and the index set of time slot is denoted as $\mathcal{T}=\{1,2,...,T\}$.
Network topology between satellites and cells is regarded as static in each epoch due to relatively short duration.
Define $x_{s,c}^f\in\{0,1\}$ to indicate the serving relationship between satellite $s$ and beam cell $c$ in epoch $f$.
$x_{s,c}^f=1$ represents that satellite $s$ serves cell $c$ in epoch $f$, and $x_{s,c}^f=0$ otherwise.
Define $y_{s,c,b}^{f,t}\in\{0,1\}$ to indicate the serving relationship between cell $c$ and the $b$-th beam of satellite $s$ in the $t$-th slot of epoch $f$.
$y_{s,c,b}^{f,t}=1$ if beam cell $c$ is served by the $b$-th beam of satellite $s$ in the $t$-th slot of epoch $f$, and $y_{s,c,b}^{f,t}=0$ otherwise.

The gateway station stores beam cell data into data queues $\mathcal{Q}_c$.
Considering that users are unevenly distributed among cells, arrival data sizes of cells are differential.
In epoch $f$, define $\alpha_{c}^f$ as the sum size of newly arrived data of all satellite users in cell $c$, and then the update process of data queue $\mathcal{Q}_c$ length can be expressed as
\begin{equation}
	Q_{c}^{f+1}=\max(Q_{c}^f-D_{c}^f,0)+ \alpha_{c}^f,
	\label{eq:queue update}
\end{equation}
where $D_{c}^f$ is the received data size of cell $c$ in epoch $f$, depending on transmission power $P_{s,c}^f$, transmission and receiving gain of antennas, bandwidth $W_{c}^f$, channel gain $h_{s,c}^f$ between the center point of cell $c$ and satellite $s$, and inter-beam interference.
Furthermore, $D_{c}^f$  can be calculated as

\begin{equation}
	D_{c}^f= \sum_{s=1}^S\sum_{b=1}^B\sum_{t=1}^T W_{c}^fT_{slot} \log (1+\frac{y_{s,c,b}^{f,t}P_{s,c}^fG_{p}h_{s,c}^f }{I_{s,c,b}^{f,t}+kT_cW_{c}^f}),
	\label{eq:rate}
\end{equation}
where  $G_{p}=G_{p}^{1}G_{p}^{2}$, $G_{p}^{1}$ and $G_{p}^{2}$ are the peak gain of transmission and receiving antennas, respectively.
$k$ is Boltzman constant, $T_c$ represents receiver temperature of cell $c$,
$I_{s,c,b}^{f,t}$ is the strength of interference the cell $c$ suffered in the $t$-th slot of epoch $f$, which is given by~\cite{multi_sat3}
\begin{equation}
	I_{s,c,b}^{f,t}=\sum_{c^\prime \in \mathcal{C}\setminus{\{c\}}} \sum_{s^\prime \in \mathcal{S}} \sum_{b^\prime \in \mathcal{B}_b} y_{s^\prime,c^\prime,b^\prime}^{f,t} P_{s,c^\prime}^fG_{p}G_{\theta_{c^\prime,s^\prime,c}^f}^{1}G_{\theta_{s^\prime,c,s}^f}^{2}h_{s^\prime,c}^f,
	\label{eq:cell_interference}
\end{equation}
where $\mathcal{B}_b$ is the set of beams with the same spectrum and polarization configuration as beam $b$.
$G_{\theta_{c^\prime,s^\prime,c}^f}^{1}$ and $G_{\theta_{s^\prime,c,s}^f}^{2}$ are the gain attenuations of transmission and receiving antennas on the direction of off-axis angles $\theta_{c^\prime,s^\prime,c}^f$ and $\theta_{s^\prime,c,s}^f$, respectively.
Particularly, $\theta_{c^\prime,s^\prime,c}^f$ can be calculated as
\begin{equation}
	\theta_{c^\prime,s^\prime,c}^f = \arccos[((d_{s^\prime,c}^f)^2+(d_{s^\prime,c^\prime}^f)^2-(d_{c,c^\prime}^f)^2)/(2d_{s^\prime,c}^fd_{s^\prime,c^\prime}^f)],
	\label{eq:theta_c_s_c}
\end{equation}
where $d_{s^\prime,c}^f$, $d_{s^\prime,c^\prime}^f$ and $d_{c,c^\prime}^f$ represent the distance between satellite $s^\prime$ and the center of cell $c$, satellite $s^\prime$ and the center of cell $c^\prime$, the centers of cells $c$ and $c^\prime$ in epoch $f$, respectively.
Similarly, $\theta_{s^\prime,c,s}^f$ can be calculated as
\begin{equation}
	\theta_{s^\prime,c,s}^f = \arccos[((d_{s,c}^f)^2+(d_{s^\prime,c}^f)^2-(d_{s,s^\prime}^f)^2)/(2d_{s,c}^fd_{s^\prime,c}^f)],
	\label{eq:theta_s_c_s}
\end{equation}
where $d_{s,s^\prime}^f$ represents the distance between satellite $s$ and satellite $s^\prime$ in epoch $f$.

Define $l_{g}^f \in [0,1]$ as the transmission load of base station $g$ in epoch $f$, which means that base station $g$ should work at least $\lceil l_{g}^fT \rceil$ time slots to guarantee the service quality of terrestrial users.
Since the coverage radius can reach several ten to hundred kilometers, a beam occupying the spectrum of terrestrial networks will incur serious co-frequency interference to numerous terrestrial cells.
Define $z_{s,c}^{f,t} \in\{0,1\}$ to indicate whether satellite $s$ occupies the spectrum of terrestrial networks to serve cell $c$.
$z_{s,c}^{f,t}=1$ means that satellite $s$ occupies the spectrum of terrestrial cells to serve beam cell $c$ in the $t$-th slot of epoch $f$, and $z_{s,c}^{f,t}=0$ otherwise.
At the center of terrestrial cell $g$, the strength of the received interference signal from the beam cell $c$ in the $t$-th slot of epoch $f$ can be expressed as~\cite{noma_spectrum_sharing1}
\begin{equation}
	I_{g,c}^{f,t} = \sum_{s \in \mathcal{S}} z_{s,c}^{f,t} P_{s,c}^fG_{p}^{1}G_{\theta_{c,s,g}^f}^{1}G_{\theta_{s,g}^f}^{3}h_{s,g}^f,
	\label{eq:satellite_terrestrial}
\end{equation}
where $G_{\theta_{c,s,g}^f}^{1}$  is the gain attenuation of satellite antenna on the direction of off-axis angle $\theta_{c,s,g}^f$, which is calculated based on the location of satellite $s$, the centers of cells $c$ and $g$.
$G_{\theta_{s,g}^f}^{3}$ is the receiving gain on the direction of off-axis angle $\theta_{s,g}^f$, and $\theta_{s,g}^f$ depends on the location of satellite $s$, cell $c$ and terrestrial users.
$h_{s,g}^f$ is channel gain between the center of cell $g$ and satellite $s$ in epoch $f$.
\begin{figure}[t]
    \centering
    \includegraphics[width=3in]{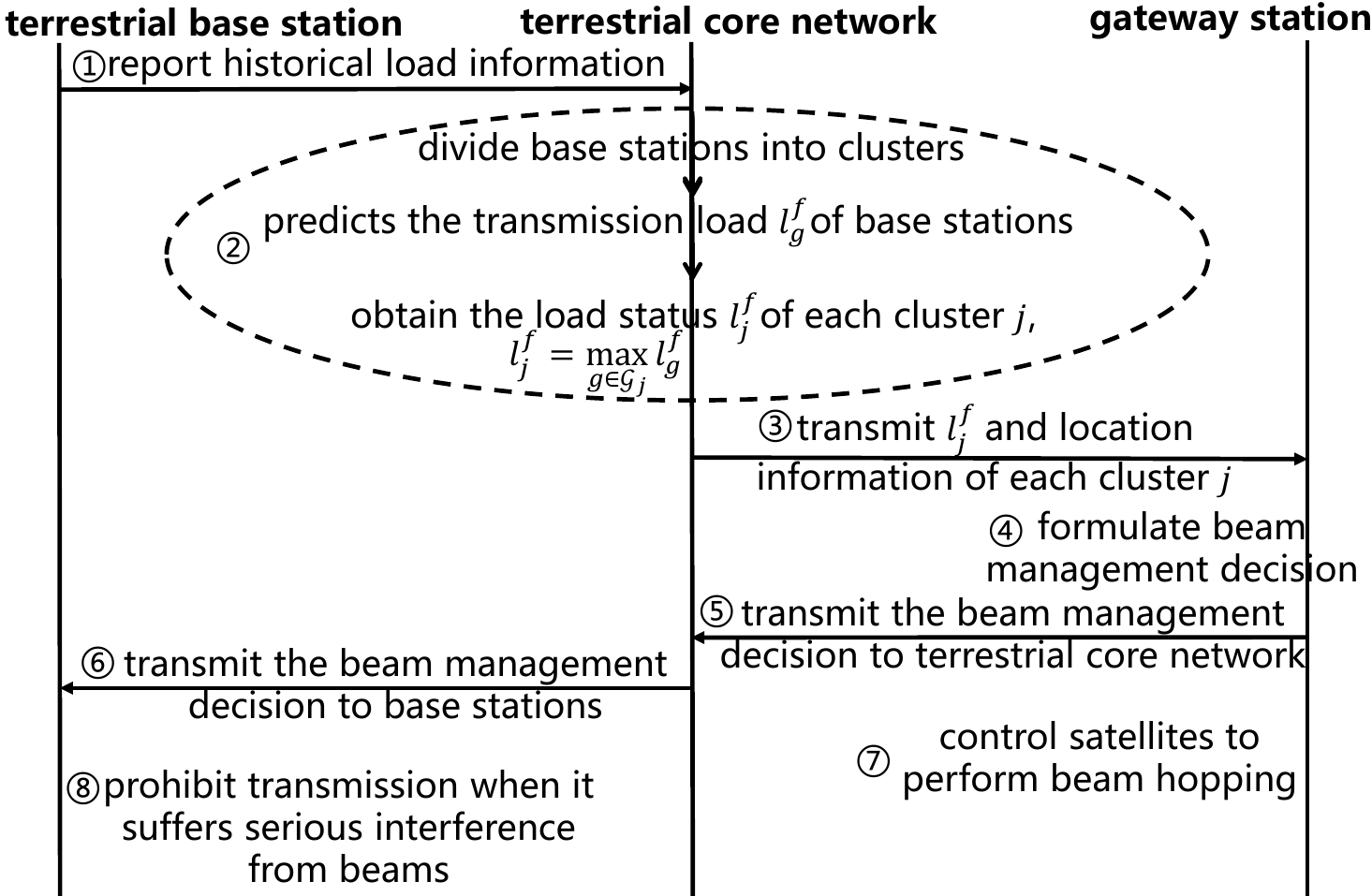}
    \caption{The process of the proposed satellite-terrestrial interference mitigation strategy to achieve spectrum sharing.}\label{fig2}
\end{figure}

\subsection{Interference mitigation strategies to achieve spectrum sharing in LEO satellite networks}\label{sec:Interference_model}

\subsubsection{Inter-beam interference mitigation strategy}
In \cite{one_sat2} and \cite{multi_sat3}, a simple interference mitigation strategy is adopted to design beam hopping patterns, where any two adjacent beam cells obtain service in different slots.
 Unfortunately, as the beam direction is far from the sub-satellite point, the beam footprint deforms, and its coverage area extends~\cite{background2}.
In this case, the serving beam brings severe interference to the nonadjacent cells of served cell, which makes the above interference mitigation strategy ineffective.
Remarkably, the gain of high directional antennas rapidly decreases when off-axis angles exceed the half-power bandwidth of antennas, and interference strength $I_{s,c,b}^{f,t}$ in equation (\ref{eq:rate}) depends on off-axis angles and interference beam number.
Hence, we propose a more appropriate inter-beam interference mitigation approach based on the off-axis angle model, whose detail is presented as follows.

Define $I_{s}^{th}$ as the acceptable maximum interference-to-noise ratio (INR) of inter-beam interference~\cite{multi_sat3}.
Denote the target signal-to-noise ratio (SNR) of cell $c$ in epoch $f$ as $SNR_{c}^f$, which can be expressed as
\begin{equation}
	SNR_{c}^f =10\log(P_{s_c,c}^fG_ph_{s_c,c}^f)-10\log(kT_cW_{c}^f),
	\label{eq:snr_express}
\end{equation}
where $s_c$ is the serving satellite of beam cell $c$.
According to equation (\ref{eq:cell_interference}) and (\ref{eq:snr_express}), the interference from cell $c^\prime$ can be calculated by $SNR_c^f$, $SNR_{c^\prime}^f$,  $h_{s_c,c}^f$, $h_{s_{c^\prime},c}^f$ and antenna gain attenuation value, where $s_{c^\prime}$ is the serving satellite of beam cell $c^\prime$.
Define $S_{c}^f$ as the maximal number of satellites that are visible to cell $c$ in epoch $f$.
Then, we can set a dynamic threshold $G_{c,c^\prime}^{th}$ for cell $c$ to restrict the maximum sum of antenna gain attenuation value caused by off-axis angles $\theta_{c^\prime,s_{c^\prime},c}$ and $\theta_{s_c,c,s_{c^\prime}}$, which can be expressed as
\begin{equation}
	G_{c,c^\prime}^{th}=I_{s}^{th}-10\log(S_{c}^fB)-SNR_{c^\prime}^f-10\log(\frac{h_{s_c,c}^f}{\delta_{c,c^\prime}h_{s_{c^\prime},c}^f}),
	\label{eq:inter_beam}
\end{equation}
where $\delta_{c,c^\prime}^f = 10^{(SNR_{c}^f-SNR_{c^\prime}^f)}$, $B$ is the maximum beam number generated by a LEO satellite.

Note that antenna gain attenuation value can be calculated based on the locations of beam cells and satellites.
Hence, we can easily judge the feasibility of two beams sharing the same spectrum, where two beams with the same polarization configuration cannot share a spectrum unless the sum of antenna gain attenuation value is less than $G_{c,c^\prime}^{th}$.

\subsubsection{Satellite-terrestrial interference mitigation strategy}

Define $I_{g}^{th}$ as the INR threshold to restrict the maximum interference between LEO satellite networks and terrestrial networks~\cite{space_spectrum_sharing1}.
Since off-axis angle $\theta_{s,g}^f$ has high randomness, we set $G_{\theta_{s,g}^f}^{3}=G_p^g$ to evaluate the worst satellite-terrestrial interference case.
In this case, INR constraint of satellite-terrestrial interference can be expressed as
\begin{equation}
	I_{g}^{f,t} =\sum_{c\in \mathcal{C}} \sum_{s \in \mathcal{S}} z_{s,c}^{f,t} P_{s,c}^fG_{p}^{1}G_{\theta_{c,s,g}^f}^{1}G_p^gh_{s,g}^f<I_{g}^{th}.
	\label{eq:inr_constraint}
\end{equation}

When beams share the spectrum of terrestrial networks, severe co-frequency interference may appear, which lowers the transmission rate and even interrupts communication links.
To improve spectral efficiency, we propose a satellite-terrestrial interference mitigation strategy to support satellite-terrestrial spectrum sharing from a spatial-temporal perspective.

As as shown in Fig.~\ref{fig2}, the process of our proposal can be concluded as eight steps.
First, terrestrial base stations report load information to terrestrial core network.
Then, the core network divides base stations into clusters to reduce the signaling overhead and the complexity of satellite-terrestrial spectrum sharing decision.
Define $\mathcal{J}=\{1,...,j,...J\}$ as the set of clusters, and the set of base stations in cluster $j$ is denoted as $\mathcal{G}_j$.
Next, the core network predicts the transmission load of each base station in epoch $f$ based on historical load information~\cite{traffic_load1,traffic_load2}, and then sets the cluster load $l_{j}^f=\max\limits_{g \in \mathcal{G}_j } l_{g}^f$.
Third, the core network provides $l_{j}^f$ and cluster locations to the gateway station.
Subsequently, the gateway station formulates beam management decisions for epoch $f$ and transmits them to the core network.
In this step, the formulated beam management decisions guarantee the transmission demand of terrestrial users, which means that the duration of strong interference a cluster suffered is less than $TT_{slot}\lfloor 1-l_{j}^f \rfloor$.
Then, the gateway station controls satellites performing beam management decisions.
Moreover, all base stations receive the interference slot information provided by the core network and then prohibit the transmission or receiving of data in these slots.

\section{Problem Formulation And Transformation}\label{sec:problem formulation}

\subsection{Problem constraints }\label{sec:Constraints}

Considering the limited capability of user equipments, a beam cell is served by a beam at most in each slot~\cite{multi_sat2}, which can be represented as
\begin{equation}
	\sum_{s=1}^S \sum_{b=1}^B y_{s,c,b}^{f,t} \leq 1, \;\; \forall c \in \mathcal C,  \forall f, t.
	\label{eq:beam_constraints1}
\end{equation}
Moreover, a beam cell is only served by a LEO satellite in each epoch, and the serving satellite must be within the beam cell's visible range, which is expressed as
\begin{equation}
	\sum_{s=1}^S \beta_{s,c}^f x_{s,c}^f =1, \;\; \forall c \in \mathcal C,  \forall f,
	\label{eq:visible_constraints}
\end{equation}
where $\beta_{s,c}^f \in \{0,1\}$.
$\beta_{s,c}^f=1$ if the elevation between satellite $s$ and the center of beam cell $c$ is larger than a preset threshold in epoch $f$, and $\beta_{s,c}^f=0$ otherwise.
In addition, a beam cannot simultaneously serve two beam cells, which is given by
\begin{equation}
	\sum_{c=1}^C y_{s,c,b}^{f,t} \leq 1, \;\; \forall f, s, t, b.
	\label{eq:beam_constraints2}
\end{equation}
Due to hardware capability limitation, a LEO satellite can generate up to $B$ beams in each slot, which can be represented as
\begin{equation}
	\sum_{c=1}^C \sum_{b=1}^{B} y_{s,c,b}^{f,t} \leq B, \;\; \forall s, t, f.
	\label{eq:beam_number_constraints2}
\end{equation}

Under the given $G_{c,c^\prime}^{th}$ in equation (\ref{eq:inter_beam}), we can obtain a tuple set $\mathcal{K}_f$, where each element  $\{s,s^\prime,c,c^\prime\} \in \mathcal{K}_f$ represents that severe interference occurs if satellites $s$ and $s^\prime$ simultaneously serve cells $c$ and $c^\prime$ via beams with the same spectrum and polarization configuration.
Then, to achieve inter-beam spectrum sharing, the inter-beam interference mitigation constraint can be expressed as
\begin{equation}
	y_{s,c,b}^{f,t} +  y_{s^\prime,c^\prime,b^\prime}^{f,t} = 0, \;\; \forall \{s,s^\prime,c,c^\prime\} \in \mathcal{K}_f, b,b^\prime \in \mathcal{B}^\prime, \forall t, f,
	\label{eq:beam_number_constraints3}
\end{equation}
where $\mathcal{B}^\prime$ is the set of beams with the same frequency band and polarization configuration.

Define $I_{s,c,j}^f$ to indicates the interference event between cluster $j$, LEO satellite $s$ and beam cell $c$ in the $t$-th slot of epoch $f$, which is given by
\begin{equation}\label{cluster_interfernence}
\begin{small}
	I_{s,c,j}^{f,t}=\left\{\begin{array}{ll}
		1, & \mathrm{if} \; (\sum_{g \in \mathcal{G}_g } \mathbb{I}(I_{g}^{f,t}>I_{g}^{th}))>0 \;\mathrm{and} \; x_{s,c}^f=1 \\
		0, & \mathrm{otherwise},
	\end{array}\right.
\end{small}
\end{equation}
where $\mathbb{I}(.)$ is indicator function.
Define $u_{j}^{f,t}\in \{0,1\}$ to indicate whether base stations in cluster $j$ suffer serious interference from LEO satellites.
$u_{j}^{f,t}=1$ if a base station in cluster $j$ suffers serious interference from LEO satellites in the $t$-th slot of epoch $f$, and $u_{j}^{f,t}=0$ otherwise.
Specifically, the relationship between $u_{j}^{f,t}$ and variable $I_{s,c,j}^f$ in (\ref{cluster_interfernence}) can be written as $u_{j}^{f,t}=\mathbb{I}(\sum_s\sum_c I_{s,c,j}^{f,t} )$.
Therefore, to support satellite-terrestrial spectrum sharing, the maximum duration constraint of satellite-terrestrial interference in epoch $f$ can be expressed as
\begin{equation}
   \sum_{t=1}^T u_{j}^{f,t} \leq T\lfloor 1-l_{j}^f \rfloor, \;\; \forall j, f.
	\label{eq:inr_constraint}
\end{equation}

Denote the tolerable maximum inter-satellite handover frequency threshold as $\overline{H}$, and then handover frequency constraint can be expressed as
\begin{equation}
   \lim_{f \to \infty} \frac{H_{c}^f}{f}  \leq \overline{H}, \;\; \forall c, f,
	\label{eq:handover_constraint}
\end{equation}
where $H_{c}^f$ is the accumulated number of inter-satellite handover of satellite users in beam cell $c$ from epochs 1 to $f$.
Moreover, we have $H_{c}^1=0$ and
\begin{equation}
   H_{c}^f=\sum_{f^\prime=2}^f \sum_{s=1}^S 1-x_{s,c}^{f^\prime-1}x_{s,c}^{f^\prime}, \;\; \forall f \geq 2.
	\label{eq:handover_cal}
\end{equation}

\subsection{Problem formulation }\label{sec:Problem Formulation}

Define $\boldsymbol{X}=(\{x_{s,c}^f\},\{y_{s,c,b}^{f,t}\},\{z_{s,c}^{f,t}\})$ as the beam management decisions in epoch $f$,
consequently, the concerned beam management optimization problem is formulated as follows.
\begin{align}
	&\boldsymbol{P_0}: \min\limits_{\{\boldsymbol{X}|\forall f\}} \lim_{F \to \infty} \frac{1}{F}\sum_{f=1}^F \sum_{c=1}^C  (D_{c}^f-Q_{c}^f)^2  \label{p0} \\
	s.t. \;& x_{s,c}^f,y_{s,c,b}^{f,t},z_{s,c}^{f,t} \in \{0,1\},  \;\; \forall c, f,t,s,b  \tag{\ref{p0}{a}} \label{p0a},\\
    & x_{s,c}^f= \mathbb{I}(\sum_{t=1}^T\sum_{b=1}^B y_{s,c,b}^{f,t}), \;\; \forall c, f,s  \tag{\ref{p0}{b}} \label{p0b},\\
    & z_{s,c}^{f,t} \leq \sum_{b=1}^B y_{s,c,b}^{f,t}, \;\; \forall c, f,t,s \tag{\ref{p0}{c}} \label{p0c},\\
	& (\ref{eq:beam_constraints1})-(\ref{eq:beam_number_constraints3}),(\ref{eq:handover_constraint}).   \notag
\end{align}
The concerned beam management problem intends to maximize the long-term service satisfaction of beam cells and control inter-satellite handover frequency.
Constraint (\ref{p0}{a}) means variables $x_{s,c}^f$, $y_{s,c,b}^{f,t}$, and $z_{s,c}^{f,t}$ are binary variables.
Constraints (\ref{p0}{b}) and (\ref{p0}{c}) represent limitions between variables $x_{s,c}^f$, $y_{s,c,b}^{f,t}$, and $z_{s,c}^{f,t}$.

\subsection{Problem transformation }\label{sec:Problem Transformation}
To cope with the complexity caused by time-averaged objective function and constraint (\ref{eq:handover_constraint}), we leverage Lyapunov drift rule to simplify the primal problem $\boldsymbol{P_0}$ into a more tractable expression \cite{lyap3,lyap2}.
Define a virtual queue $\mathcal{M}_{c}$ for beam cell $c$ to characterize the relationship between inter-satellite handover frequency and handover threshold $\overline{H}$ \cite{lyap3}, whose length in epoch $f$ is denoted as $M_{c,f}$.
We have
\begin{equation}
   M_{c,f}=\max (M_{c,f-1}+m_{c,f},0), \forall c,
	\label{eq:handover_queue}
\end{equation}
where $m_{c,f}=(1-\sum_{s=1}^S x_{s,c}^{f-1}x_{s,c}^f)-\overline{H}$  and $M_{c,1}=0$.
\begin{lemma}
The inter-satellite handover frequency constraints (\ref{eq:handover_constraint}) can be replaced by queue stabilization constraints.
\end{lemma}

\begin{proof}
Based on queue update process in (\ref{eq:handover_queue}), we have
\begin{equation}
   M_{c,f}-M_{c,f-1} \geq m_{c,f}.
	\label{eq:handover_queue_1}
\end{equation}
Summing over $f\in \{2,3,...,F\}$, we can obtain
\begin{equation}
	\begin{aligned}
   M_{c,F} - M_{c,1} &\geq \sum_{f=1}^F m_{c,f}= H_{c}^F - F\overline{H}.
	\label{eq:handover_queue_2}
    \end{aligned}
\end{equation}
Dividing by $F$ and using $M_{c,1}=0$, it can prove that queue  $\boldsymbol{M}_{c}$ stabilization is equivalent to the concerned maximum inter-satellite handover frequency constraint (\ref{eq:handover_constraint}) being satisfied.
\end{proof}

Define $\boldsymbol{M}_f=(M_{1,f},...,M_{C,f})$ as the vector of all virtual queue $M_{c,f}$, and denote the Lyapunov function $L(\boldsymbol{M}_f)= \frac{1}{2} \sum_{c=1}^C M_{c,f}^2$.
Then, denote the drift in the Lyapunov function as $\Delta(f)$, and it is given by $\Delta(f) \triangleq L(\boldsymbol{M}_{f+1})-L(\boldsymbol{M}_f)$.
Define an auxiliary variable $\delta_f$, which is given by
\begin{equation}
   \delta_f = V[\sum_{c=1}^C  (D_{c}^f-Q_{c}^f)^2] +\sum_{c=1}^C M_{c,f}m_{c,f},
   \label{eq:new_object}
\end{equation}
where $V \geq 0$ is a constant.
\begin{lemma}
The objective function of problem $\boldsymbol{P}_0$ can be converted into $\delta_f$, and the desired algorithm minimizing $ \delta_f$ can guarantee the queue stability of $\mathcal{M}_{c}$.
\end{lemma}

\begin{proof}
To solve the minimization problem $\boldsymbol{P}_0$, we apply the drift-plus-penalty technology and then obtain the equivalent objective function as follows \cite{lyap2}
\begin{equation}
   \min \Delta(f) + V E[ \sum_{c=1}^C  (D_{c}^f-Q_{c}^f)^2 |\boldsymbol{M}_f ].
   \label{eq:new_object2}
\end{equation}

Under any beam management decision $\boldsymbol{X}$, we can obtain the upper bound of (\ref{eq:new_object2}) as follows \cite{lyap3}
\begin{equation}
   \begin{aligned}
   \Delta(f) +& V E[\sum_{c=1}^C  (D_{c}^f-Q_{c}^f)^2 |\boldsymbol{M}_f ] \leq B_0 \\
   &+ E[ \sum_{c=1}^C M_{c,f}m_{c,f}+ V (D_{c}^f-Q_{c}^f)^2 |\boldsymbol{M}_f ],
   \label{ie:Lyapunov2}
   \end{aligned}
\end{equation}
where $B_0$ is a constant and $  B_0 \geq \frac{1}{2} \sum_{c=1}^C  E[ m_{c,f}^2 | L(\boldsymbol{M}_f)].$

In this case, the minimization of (\ref{eq:new_object2}) can be achieved by minimizing the right-hand-side of inequality (\ref{ie:Lyapunov2})\cite{lyap3}, \cite{lyap2}, which means that problem $\boldsymbol{P}_0$ can be solved by greedily minimizing the following objective function in each epoch:
\begin{equation}
    E[ \sum_{c=1}^C M_{c,f}m_{c,f}+ V (D_{c}^f-Q_{c}^f)^2 |\boldsymbol{M}_f ]= E[ \delta_f |\boldsymbol{M}_f ].
\end{equation}
Hence, the conversion of the objective function is proven.
Next, we prove the stability of queue $\mathcal{M}_c$.
Since $m_{c,f} \leq 1-\overline H$ and $M_{c,f}\leq 2(1-\overline H)$, it can be inferred that
\begin{equation}
  \frac{1}{2} \sum_{c=1}^C M_{c,f}^2 -M_{c,f-1}^2 \leq B_0 + 2C(1-\overline H)^2.
   \label{ie:Lyapunov3}
\end{equation}
Summing over $f\in \{2,3,...,F\}$, we have $E(M_{c,F}^2)\leq 2F(B_0 +2C(1-\overline H)^2)$.
Due to $E(M_{c,F})^2\leq E(M_{c,F}^2)$, we have
\begin{equation}
  E(M_{c,F}) \leq \sqrt{2F(B_0 + 2C(1-\overline H)^2)}.
   \label{ie:Lyapunov4}
\end{equation}
Finally, it can be obtained that $\lim_{F \to \infty} \frac{M_{c,F}}{F}=0$.
\end{proof}

Based on the above analysis, the primal problem $\boldsymbol{P}_0$ can be equivalently transformed into a series of per-epoch problems $\boldsymbol{P}_1$ as follows
\begin{align}
	\boldsymbol{P_1}: & \min\limits_{\{\{x_{s,c}^f\},\{y_{s,c,b}^{f,t}\},\{z_{s,c}^{f,t}\} | \boldsymbol{M}_f,f \}} \delta_f  \label{p1}\\
	s.t. \;&(\ref{eq:beam_constraints1})-(\ref{eq:beam_number_constraints3}), (\ref{p0}{a})-(\ref{p0}{c}).  \notag
\end{align}

\section{Problem Decomposition and Beam Management Algorithm Design}\label{sec:Algorithm Framework}

\subsection{Problem decomposition}\label{sec:problem_decomposition}

When slot number in each epoch is 1, problem $\boldsymbol{P_1}$ degenerates into a multiple-knapsack problem, which is NP-hard~\cite{NP-hard}.
In addition, due to enormous amounts of variables and constraint terms, current iteration based methods and solvers inefficiently obtain a high-quality solution within a reasonable time overhead.
For these reasons, this part further decouples problem $\boldsymbol{P_1}$ into three tractable sub-problems.
The problem decomposition framework is shown in Fig.~\ref{fig4}, and corresponding reasons are concluded as follows.

Inter-satellite handover control is a core aspect of beam management approach, which requires the gateway station preassigning serving satellite beams to cells and maintaining a relatively low handover frequency.
If the developed handover algorithm is executed in each epoch or iteratively invoked, it will significantly increase the inter-satellite handover frequency~\cite{min_distance}.
Hence, we first decompose an inter-satellite handover decision problem from problem $\boldsymbol{P}_1$, and then design a conditional handover triggering mechanism.
The beam hopping design is independent of satellite-terrestrial spectrum sharing decisions in our scenario.
Moreover, satellite-terrestrial spectrum sharing cooperation may be canceled when LEO satellite networks work on a low load state or terrestrial networks overload.
Hence, under the given serving satellites of beam cells, problem $\boldsymbol{P}_1$ are further divided into beam hopping design problem and satellite-terrestrial spectrum sharing problem.

Considering the requirement of practical applications, three decomposed problems are only solved once in each epoch.
In addition, the required decision space sizes of decomposed subproblems are far less than problem $\boldsymbol{P_1}$, which contributes to obtaining desired beam management decisions with low complexity.
\begin{figure}[!t]
	\center
	\includegraphics[width=3.1in]{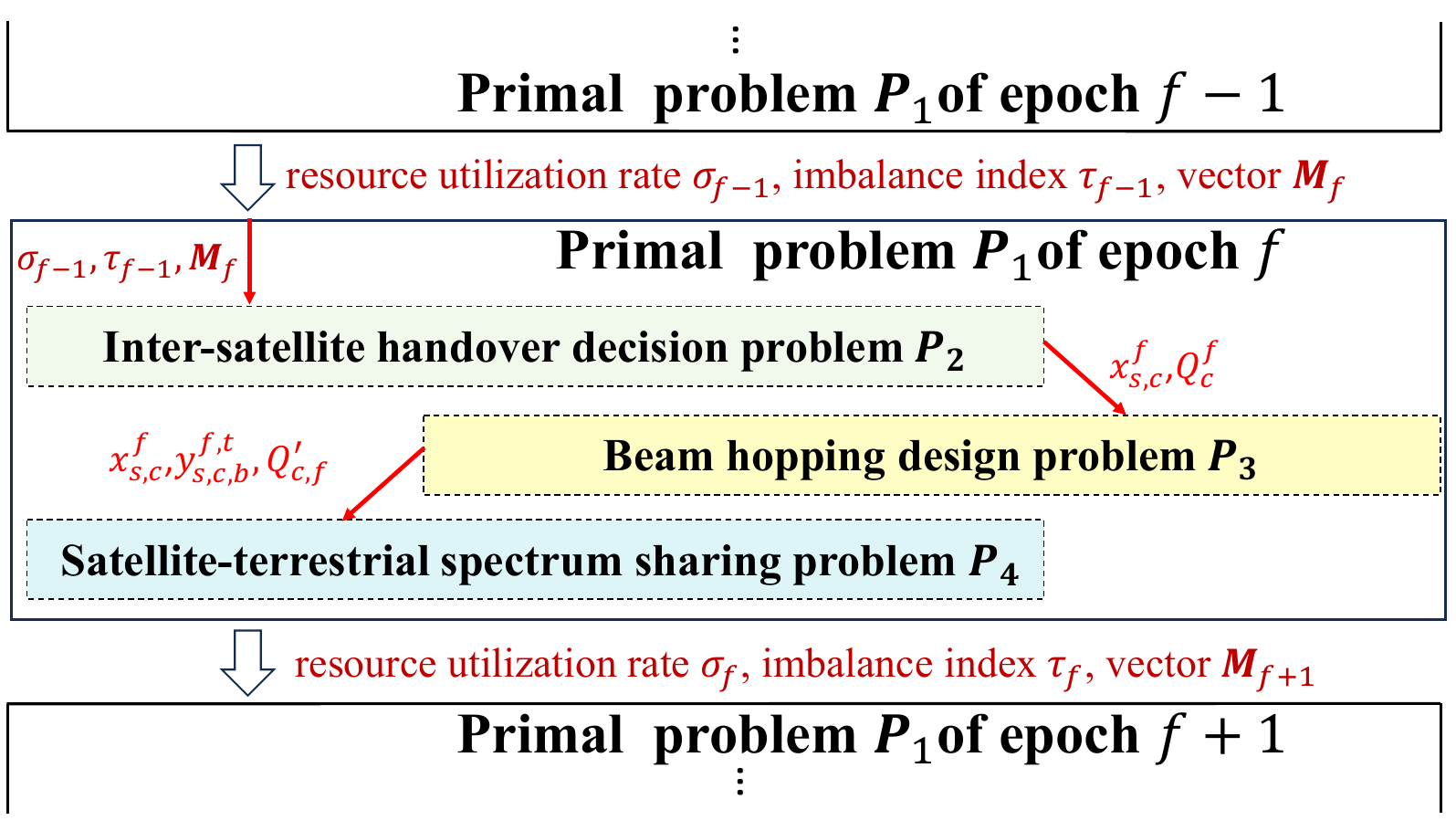}
	\caption{Beam management problem decomposition framework.}\label{fig4}
\end{figure}
\subsection{Inter-satellite handover decision problem and algorithm design}\label{sec:subproblem 1}

Suppose several hop-spot areas are served by a satellite and other low-demand beam cells are allocated to another satellite.
In this case, a satellite may overload while another satellite is leisure, which degrades network capacity.
Thereby, load balance strategy plays a crucial role in LEO satellite networks to improve capacity and fairness among beam cells \cite{multi_sat2}.
Particularly, we need to pre-assign serving satellites to beam cells before allocating time-frequency transmission resource, and the objective function value of problem $\boldsymbol{P_1}$ relies on beam hopping design and satellite-terrestrial spectrum sharing decisions.
Hence, it is essential to adjust the objective function of inter-satellite handover decision problem based on the primal objective function of the problem $\boldsymbol{P_1}$.
To this end, an inter-satellite handover decision problem based on problem $\boldsymbol{P_1}$ is formulated as follows
\begin{align}
	\boldsymbol{P_2}: & \min\limits_{\{\{x_{s,c}^f\} | \boldsymbol{M}_f,f \}} \delta_f^\prime  \label{p2}\\
	s.t. \;&(\ref{eq:visible_constraints}), (\ref{p0}{a}).  \notag
\end{align}
where
\begin{equation}
   \delta_f^\prime= \sum_{s=1}^S (\frac{\sum_{c=1}^C (x_{s,c}^f-\frac{1}{2})Q_{c}^f}{\sum_{c=1}^C Q_{c}^f})^2 + \sum_{c=1}^C M_{c,f}m_{c,f},
   \label{eq:delta_f-prime}
\end{equation}
and the first term of the right-hand-side of equation (\ref{eq:delta_f-prime}) represents the normalized weight value of satellite loads.

Actually, LEO satellite load fluctuates over time due to random traffic arrival.
If load balance strategy is performed in each epoch, inter-satellite handover frequency inevitably increases without apparent benefits.
To overcome this issue, we propose a conditional handover triggering mechanism, where inter-satellite handover procedure is only executed if one of conditions holds.
Note that inter-satellite handover procedure must be implemented when serving satellite leaves the visible area of a served cell.
In this case, the gateway station requires reallocating a satellite to this beam cell.
In addition, due to the movement of satellites, beam hopping design and satellite-terrestrial spectrum sharing algorithms may output low-quality solutions as interference situations change.
The low-quality solution represents that LEO satellite networks simultaneously obtain a low resource utilization rate and a load imbalance index.
Specifically, the resource utilization rate is the average ratio of the occupied time-frequency-space resource and all available resources of all beams, and the load imbalance index is the ratio of maximum load and minimum load among satellites.
Therefore, the gateway station is required to monitor load statuses and resource utilization rates of LEO satellite networks.
If resource utilization rate and imbalance index are simultaneously lower than preset thresholds $\sigma_0$ and $\tau_0$, gateway station will automatically reallocate serving satellites to beam cells.

Based the proposed conditional handover triggering mechanism, we further develop a proactive algorithm to solve problem $\boldsymbol{P_2}$, which is shown in Algorithm~\ref{alg:handover_Algorith}.
First, beam cells without serving satellite will obtain a new serving satellite.
and the initial serving satellites of them are obtained by a widely used multi-attribute handover decision scheme based on entropy in~\cite{multi-attribute}, where satellite loads, remaining service time, and elevation angle are considered as attributes (line 2-14).
Unfortunately, the above stage tends to load imbalance among satellites since it will distribute cells to satellites with high loads.
Therefore, we further control load distributions and inter-satellite handover frequency based on matching theory (line 15-36).
To enhance the exploration capacity of the proposal, we randomly adjust serving satellites of several cells (line 16).
Moreover, the expected value of the objective function $\delta_f^\prime $ monotonically decreases by swapping matching operations, which ensures convergence.

Define $\mathcal{C}_s$ as the set of beam cells served by satellite $s$, and $\mathcal{C}_s=\{c | x_{s,c}^f=1\}$.
Thereby, we have $\mathcal{C}_{s_1} \cap \mathcal{C}_{s_2}= \varnothing,  \forall s_1, s_2 \in \mathcal{S}$ and $\cup_{s=1,...,S} \mathcal{C}_{s} = \mathcal{C}$.
Define $\boldsymbol{\mathcal{C}_s}=(\mathcal{C}_1, \mathcal{C}_2,...\mathcal{C}_S)$, and then denote the objective function value $\delta_f^\prime $ under decision $\boldsymbol{\mathcal{C}_s}$ as $\delta_f^\prime(\boldsymbol{\mathcal{C}_s})$.
For any two beam cells $c_1 \in \mathcal{C}_{s_1}$ and $c_2 \in \mathcal{C}_{s_2}$, there are two kinds of swap operations for them.
The first kind of swap operation converts $\mathcal{C}_{s_1}$ and $ \mathcal{C}_{s_2}$ into $\mathcal{C}_{s_1}^\prime$ and $ \mathcal{C}_{s_2}^\prime $ by removing and adding beam cell operations (line 19-26),  $\mathcal{C}_{s_1}^\prime=\mathcal{C}_{s_1} \setminus \{ c_1\} \cup \{c_2\}$ and $\mathcal{C}_{s_2}^\prime=\mathcal{C}_{s_2} \setminus \{ c_2\} \cup  \{c_1\}$.
The second kind of swap operation converts $\mathcal{C}_{s_1}$ and $ \mathcal{C}_{s_2}$ into $\mathcal{C}_{s_1}^\prime$ and $ \mathcal{C}_{s_2}^\prime $ by removing a cell from $\mathcal{C}_{s_1}^\prime$ and then adding it into $\mathcal{C}_{s_2}^\prime$ (line 27-34), i.e., $\mathcal{C}_{s_1}^\prime=\mathcal{C}_{s_1} \setminus \{ c_1\} $ and $\mathcal{C}_{s_2}^\prime=\mathcal{C}_{s_2} \cup \{c_1\}$.
The resulted vector after once swap operation is denoted as $\boldsymbol{\mathcal{C}_s^\prime}=(\mathcal{C}_1,...,\mathcal{C}_{s_1}^\prime,\mathcal{C}_{s_2}^\prime,...,\mathcal{C}_S)$.
In swap matching stage, a swap is accepted only if $\mathcal{C}_{s_1}^\prime$, $\mathcal{C}_{s_2}^\prime$ satisfy constraints (\ref{eq:visible_constraints}) and $\delta_f^\prime(\boldsymbol{\mathcal{C}_s^\prime}) < \delta_f^\prime(\boldsymbol{\mathcal{C}_s})$~\cite{one_sat2}.
Finally, Algorithm~\ref{alg:handover_Algorith} outputs serving relationship $\{x_{s,c}^f\}$ when it cannot find a feasible swap operation or the number of iterations reaches threshold $N^\prime$.

\begin{algorithm} [t]
\begin{algorithmic}[1]
    \footnotesize
	\caption{Inter-satellite Handover Decision Algorithm}
	\label{alg:handover_Algorith}
	\State Input: $\boldsymbol{M}_f$, $Q_{c}^f$, $\overline{H}$, $x_{s,c}^{f-1}$, $x_{s,c}^f= 0$, $\sigma_f$, $\tau_f$.
	\For{$c=1:C$}
	\If{ the serving satellite can continue to serve cell $c$ }
	\State $x_{s,c}^f=x_{s,c}^{f-1}$.
    \State Update the load of the serving satellite of cell $c$.
    \EndIf
    \EndFor
    \For{$c=1:C$}
	\If{ $x_{s,c}^f==0$ }
	\State Allocate the serving satellite by using the multi-attribute decision handover scheme based on entropy in~\cite{multi-attribute}.
    \State $\sigma_f=0$, $\tau_f=0$.
    \State Update the load of the serving satellite of cell $c$.
    \EndIf
    \EndFor
    \If{$\sigma_f < \sigma_{0}$ and $\tau_f<\tau_0$}
    \State Randomly adjust serving satellites of several cells.
    \Repeat
    \State Construct $\boldsymbol{\mathcal{C}_s}$ and calculate  $\delta_f^\prime(\boldsymbol{\mathcal{C}_s})$ based on current $\{x_{s,c}^f\}$.
    \For{$c_1=1:C$}
    \For{$c_2=1:C$}
    \State Select a first kind swap and obtain $\delta_f^\prime(\boldsymbol{\mathcal{C}_s^\prime})$.
    \If {$\delta_f^\prime(\boldsymbol{\mathcal{C}_s^\prime}) < \delta_f^\prime(\boldsymbol{\mathcal{C}_s})$}
    \State Update $x_{s,c}^f$ and $\delta_f^\prime(\boldsymbol{\mathcal{C}_s})$.
    \EndIf
    \EndFor
    \EndFor
    \For{$c=1:C$}
    \For{$s=1:S$}
    \State Select a second kind swap and obtain $\delta_f^\prime(\boldsymbol{\mathcal{C}_s^\prime})$.
    \If {$\delta_f^\prime(\boldsymbol{\mathcal{C}_s^\prime})< \delta_f^\prime(\boldsymbol{\mathcal{C}_s})$}
    \State Update $x_{s,c}^f$ and $\delta_f^\prime(\boldsymbol{\mathcal{C}_s})$.
    \EndIf
    \EndFor
    \EndFor

    \Until{the number of iterations reaches $N^\prime$ or all $\delta_f^\prime(\boldsymbol{\mathcal{C}_s^\prime}) \leq \delta_f^\prime(\boldsymbol{\mathcal{C}_s})$}
    \EndIf
	\State Output:  $\{x_{s,c}^f\}$.
	\end{algorithmic}
\end{algorithm}

\subsection{Beam hopping design problem and algorithm design}\label{sec:subproblem 2}

Under the given serving relationship $x_{s,c}^f$, inter-satellite handover frequencies of beam cells are deterministic values.
Therefore, the objective function $\delta_f$ of problem $\boldsymbol{P}_1$ degrades into $\sum_{c=1}^C  (D_{c}^f-Q_{c}^f)^2$, whose solution dimensions are proportional to the number of available beams, beam cells, and time slots in each epoch.
Since such large solution dimensions significantly increase the complexity of obtaining a proper beam hopping design~\cite{one_sat2}, we solve problem $\boldsymbol{P}_1$ by greedily minimizing $\sum_{c=1}^C  (D_{c}^f-Q_{c}^f)^2$ in each slot in sequence.
Define $D_{c}^{f,t}$ as the transmission data size of cell $c$ in the $t$-th slot of epoch $f$ via the exclusive spectrum of LEO satellite networks, which can be calculated as
\begin{equation}\label{cell_data_slot}
	D_{c}^{f,t}=\sum_{s=1}^C \sum_{b=1}^B y_{s,c,b}^{f,t}W_1T_{slot}\log(1+SNR_{c}^f).
\end{equation}
Denote the remaining data queue length of beam cell $c$ of the $t$-th slot of epoch $f$ as $Q_{c}^{f,t}$, where $Q_{c}^{f,t+1}=\max(0,Q_{c}^{f,t}-D_{c}^{f,t})$ and $Q_{c}^{f,1}=Q_{c}^f$.
Then, the beam hopping design problem in slot $t$ can be expressed as
\begin{align}
	\boldsymbol{P_3}: & \min\limits_{\{\{y_{s,c,b}^{f,t}\} | x_{s,c}^f,Q_{c}^{f,t},f,t \}} \sum_{c=1}^C  (D_{c}^{f,t}-Q_{c}^{f,t})^2  \label{p3}\\
	s.t. \;&  (\ref{eq:beam_constraints1})-(\ref{eq:beam_number_constraints3}), (\ref{p0}{a}), (\ref{p0}{b}).  \notag
\end{align}

Since the inter-beam interference situation can be regarded as fixed in each epoch with a short duration, we can construct a weighted conflict graph for all slots in an epoch.
Define $\mathcal{V}_f$ as the vertex set of conflict graph in epoch $f$, and a vertex represents a feasible $y_{s,c,b}^{f,t}$ of a cell.
There are bound to $B$ vertices corresponding to a cell, representing that a cell can be served by any beam of the serving satellite.
Define  $\mathcal{E}_f$ as the edge set in conflict graph.
If an edge connects two vertices $v_1$ and $v_2$, it indicates that two beam hopping decisions $y_{s_1,c_1,b_1}^{f,t}$, $y_{s_2,c_2,b_2}^{f,t}$ corresponding to vertices $v_1$ and $v_2$ violate a constraint term.
Remarkably, the situations in which an edge appears can be summarized into three categories.
\begin{itemize}	
	\item Two vertices represent the same cell, which violates constraint (\ref{eq:beam_constraints1}).
	\item Two vertices represent that two cells are simultaneously served by a beam, which violates constraint (\ref{eq:beam_constraints2}).
    \item Severe inter-beam interference occurs, which violates constraint (\ref{eq:beam_number_constraints3}).	
\end{itemize}

Fig.~\ref{fig3} presents an example of a LEO satellite network scenario and corresponding conflict graph, where satellite 1 serves cells 1 and 2, and satellite 2 serves cells 3 and 4.
Beam 1 and 2 operate on different antenna polarization methods and the number of vertices in set $\mathcal{V}_f$ is 8.
An edge connects vertices 1 and 2 since they all represent cell 1.
Severe inter-beam interference occurs when cells 1 and 3 simultaneously obtain service from beams with the same polarization configuration.
Hence, vertex 1 directly connects with vertex 5, and vertex 2 connects with vertex 6. 

\begin{figure}[!t]
	\center
	\includegraphics[width=3in]{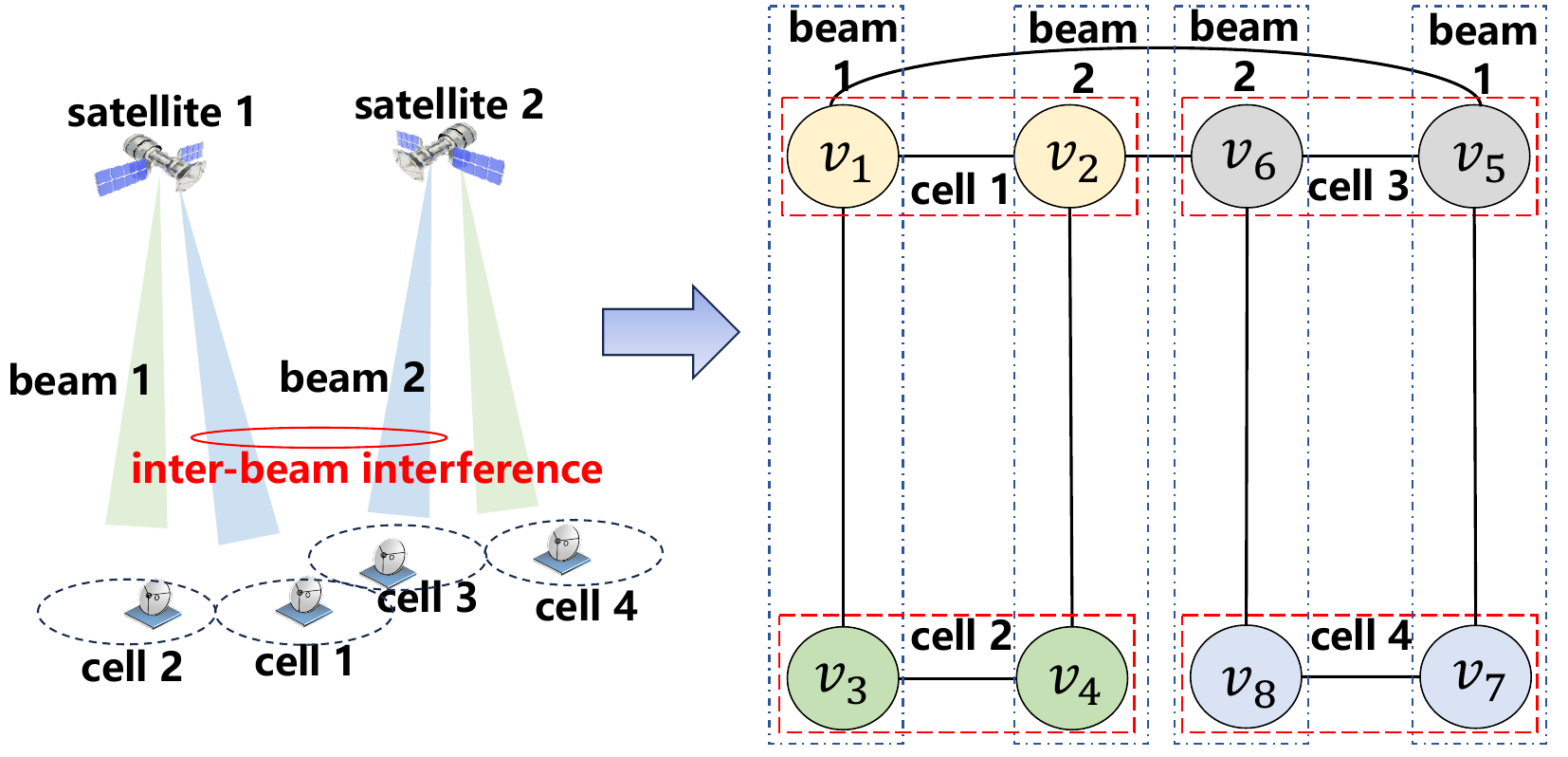}
	\caption{An example of a network scenario and corresponding conflict graph. }\label{fig3}
\end{figure}

The weight of vertex $v$ corresponding to beam cell $c$ is calculated by a subtraction operation between a constant and term $(D_{c}^{f,t}-Q_{c}^{f,t})^2$, which is given by
\begin{equation}\label{eq:cell_weight}
\begin{small}
	w_v= (W_1T_{slot}\log(1+SNR_{c}^f))^2+(Q_{c}^{f,t})^2-(D_{c}^{f,t}-Q_{c}^{f,t})^2.
\end{small}
\end{equation}
The weighted maximum independent set (WMIS) problem aims to search a subset of vertices with the maximum weight sum, where any two vertices in subset are disconnected~\cite{NP-hard2}.
Hence, vertices in WMIS satisfy all constraints of problem $\boldsymbol{P_3}$.
In addition, the maximum weight sum of WMIS corresponds to the minimum objective function value of problem $\boldsymbol{P_3}$.
Therefore, problem $\boldsymbol{P_3}$ can be converted into a WMIS problem based on the built conflict graph.

\begin{algorithm} [htbp]
\begin{algorithmic}[1]
    \footnotesize
	\caption{Beam Hopping Design Algorithm}
	\label{alg:msia_Algorithm}
	\State Input: $x_{s,c}^f$, $T$, $Q_{c}^f$, $\mathcal Result = \emptyset$.
    \State Construct a conflict graph $(\mathcal{V}_f, \mathcal{E}_f)$, $Q_{c}^{f,1}=Q_{c,f}^\prime=Q_{c}^f$.
    \For{t=1:T}
       \State Set available state $\boldsymbol{f}=\boldsymbol{0}_{1\times \mathcal{V}_f|}$ and $\mathcal{V}^\prime = \emptyset$.
       \State Set an inaccessible state $\boldsymbol{f}(k)=1$ for all vertices with empty queue.
       \State Calculate the weight $w_v$ and weight radio $\rho_v$ for all vertices.
       \State According to descending order of weight radio $\rho_v$, visit all vertices.
       \For{$k=1:|\mathcal{V}_f|$}
          \If{$ \boldsymbol{f}(k)==0$}
            \State Find the vertex $v_k$ corresponding to the sorted index $k$.
            \If{$\mathcal{V}^\prime == \emptyset$   or edge $ (v_k, v_k^\prime)  \not\in \mathcal{E}_f,\; \forall v_k^\prime \in \mathcal{V}^\prime $}
            \State $\mathcal Result \gets \mathcal Result \cup \{y_{s,c,b}^{f,t}\}$.
	        \State $\mathcal{V}^\prime\gets \mathcal{V}^\prime \cup \{v_k\}$,
            \State Set inaccessible states $\boldsymbol{f}(k)=1$ for all vertices in $\mathcal{V}_{v_k}$.
\EndIf
            \State Update $Q_{c}^{f,t}$ and set $Q_{c,f}^\prime=Q_{c}^{f,t}$.
            \EndIf
       \EndFor
    \EndFor
	\State Output:  $\mathcal Result$.
	\end{algorithmic}
\end{algorithm}

As shown in Algorithm~\ref{alg:msia_Algorithm}, we propose a low complexity greedy search algorithm to solve the transformed WMIS problem.
We first set inaccessible states for vertices with empty data queue (line 4-5), and then calculate weight ratios $\rho_v$ for rest accessible vertices, where $ \rho_v =w_v/(w_v+\sum_{v^\prime \in \mathcal{V}_v } w_{v^\prime})$, and $\mathcal{V}_v$ is the set of accessible vertices adjacent to vertex $v$ (line 6).
Next, accessible vertices are visited according to the descending order of weight ratio $\rho_v$ (line 7-18).
In this stage, when an accessible vertex is recorded, its adjacent vertices are set as inaccessible states (line 11-15).
Meanwhile, we calculate the $Q_{c}^{f,t}$ of cells obtaining service (line 16).
When all vertices are set to inaccessible states, we update the weight and accessibility of vertices and then repeat the above process for next slot.
Finally, Algorithm~\ref{alg:msia_Algorithm} outputs beam hopping design decisions $\{y_{s,c,b}^{f,t}\}$  and the rest data queue length $Q_{c,f}^\prime$ of cells.

\subsection{Satellite-terrestrial spectrum sharing problem and algorithm design}\label{sec:subproblem 3}

Define $D_{c,f}^\prime$ as the transmission data size of beam cell $c$ in epoch $f$ via the spectrum of terrestrial networks, and it can be calculated as
\begin{equation}\label{cell_data_slot}
\begin{small}
	D_{c,f}^\prime=\sum_{s=1}^S \sum_{t=1}^T z_{s,c}^{f,t}W_2T_{slot}\log(1+SNR_{c}^f).
\end{small}
\end{equation}

In order to achieve high spectral efficiency and minimize the objective function of problem $\boldsymbol{P}_1$, an optimization problem $\boldsymbol{P_4}$ based on the given serving satellite and beam hopping design is formulated as follows.
\begin{align}
 	\boldsymbol{P_4}:	& \min\limits_{ \{\{z_{s,c}^{f,t}\} |Q_{c,f}^\prime, x_{s,c}^f, y_{s,c,b}^{f,t},f \}}  \quad	  \sum_{c=1}^C (D_{c,f}^\prime-Q_{c,f}^\prime)^2 \label{spectrum_sharing}\\
	s.t. \; \; &(\ref{eq:inr_constraint}), (\ref{p0}{a}), (\ref{p0}{c}). \notag
\end{align}

Although problem $\boldsymbol{P}_4$ can be converted into an integer programming problem, it confronts an unacceptable computational time cost in directly applying the classical branch and bound method or solvers~\cite{Yuan}.
To address this issue, we develop a low complexity satellite-terrestrial spectrum sharing algorithm to solve problem $\boldsymbol{P}_4$, which combines with improved binary sparrow search algorithm in \cite{sparrow_search} and greedy search algorithm, whose detail is summarized in Algorithm~\ref{alg:spectrum_sharing_Algorith}.

Firstly, decision variables $z_{s,c}^{f,t}$ of beam cells with an empty queue are set to 0, and the dimension of a sparrow equals the number of rest variables $z_{s,c}^{f,t}$, which contributes to reducing the decision space size.
Since chaotic mapping has powerful capacities in increasing exploration capacity,
we adopt tent chaotic strategy to initialize sparrows with continuous values between 0 and 1.
Define $S_{i,j}$ as the value of the $j$-th continuous variable of the $i$-th sparrow, and corresponding binary value $S_{i,j}^\prime$ is obtained by inputting $S_{i,j}$ to a s-shape function, which is given by~\cite{sparrow_search2}
\begin{equation}\label{binary_operation}
	S_{i,j}^\prime=\left\{\begin{array}{ll}
		1, &if \; \frac{1}{1+e^{-2S_{i,j}}} > \mu, \\
		0, &otherwise,
	\end{array}\right.
\end{equation}
where $\mu$ is a random value between 0 and 1.
Moreover, the quality of sparrow $i$ is evaluated by its fitness value $F_i$, which can be calculated as
\begin{equation}\label{fitness}
\begin{small}
	F_i=\left\{\begin{array}{ll}
		\sum_{c=1}^C \Omega_c-(D_{c,f}^\prime-Q_{c,f}^\prime)^2, &if \; constraint \; (\ref{eq:inr_constraint}) \;holds, \\
		\sum_{c=1}^C \Omega_c -(Q_{c,f}^\prime)^2, &otherwise,
	\end{array}\right.
\end{small}
\end{equation}
where $\Omega_c = {Q_{c,f}^\prime}^2+(W_2T_{slot}\log(1+SNR_{c}^f)T)^2$.

Next, we sort sparrows according to the descent order of fitness value and then divide them into producers and scroungers.
Meanwhile, several sparrows are randomly selected as spectators.
The update processes of producers, scroungers, and spectators are consistent with the traditional sparrow search algorithm in \cite{sparrow_search}.
Moreover, our proposal adopts additional local search  (line 7-22) and adaptive crossover (line 23-30) steps to enhance the exploit capacity.
Local search is executed by performing mutation operations for the sparrow with maximum fitness value.
Adaptive crossover procedure is worked on the half of sparrows with worse fitness value.
After the improved sparrow search stage, we subsequently employ a greedy search stage to increase spectral efficiency (line 35-37).
In this stage, we set $z_{s,c}^{f,t} =1$ if constraint (\ref{eq:inr_constraint}) still holds after performing value mutation operation on $z_{s,c}^{f,t}$.

\begin{algorithm}[t]
	\begin{algorithmic}[1]
		\caption{Satellite-Terrestrial Spectrum Sharing Algorithm}
		\label{alg:spectrum_sharing_Algorith}
		\footnotesize
        \State Set $z_{s,c}^{f,t}=0$ for all cells with an empty queue $Q_{c,f}^\prime$.
		\State Input: $x_{s,c}^f$, $y_{s,c,b}^{f,t}$, the maximum number of iterations $N^{\prime \prime}$, the number of local search iterations $N^{\prime \prime \prime}$, population size $N_{pop}$, sparrow size $Z_s$, the number of producers $P_d$, the number of spectators $S_d$.
		\State Initialize sparrow population by using tent chaotic strategy.
        \State Obtain binary $z_{s,c}^{f,t}$ based on (\ref{binary_operation}) and record the global binary optimal solution $S_g$ and its fitness $f(S_g)$.
		\While {$i<N^{\prime \prime}$}

        \State Sort sparrows according to the descent order of fitness, find the current best individual $S_b$.

        \While {$j<N^{\prime \prime \prime}$}
        \State $S_b^\prime=S_b$.
        \For{$j=1:4$}
        \State Randomly select a elements $s_b$ in the best sparrow $S_b^\prime$.
        \State $s_b$ = 1- $s_b$.
        \EndFor
        \State Obtain the binary solution $z_{s,c}^{f,t}$ of $S_b^\prime$ based on (\ref{binary_operation}).
        \State Calculate fitness $f(S_b^\prime)$.
        \If {$f(S_b^\prime)>f(S_b)$}
        \State $S_b=S_b^\prime$.
        \If {$f(S_b^\prime)>f(S_g)$}
        \State Update global binary optimal solution $S_g$.
        \State $f(S_g)=f(S_b^\prime)$.
        \EndIf
        \EndIf
        \EndWhile

        \State Calculate $\upsilon = 0.55-\frac{0.1}{1+\exp(5-10*i/N^{\prime\prime})}$
        \For{$j=N_{pop}/2+1 :N_{pop}$}.
        \If{$rand(1) <\upsilon$}
        \State Generate a randomly integer $\zeta$ in $\{1,2,3\}$.
        \State Randomly select $\zeta$ elements in sparrow $j$.
        \State For each selected element $s_j(\zeta)$, we have  $s_j(\zeta)= 1- s_j(\zeta)$.
        \EndIf
        \EndFor
        \State Update all sparrows' location according to \cite{sparrow_search}.
        \State Obtain the binary solution of all sparrow based on (\ref{binary_operation}).
        \State If the fitness of new sparrows are better than $f(S_g)$, update $f(S_g)$ and global binary optimal solution $S_g$.
       \EndWhile
       \For{$i=1:Z_s$}
       \State Set the $i$-th $z_{s,c}^{f,t}$ as 1 if all constraint (\ref{eq:inr_constraint}) can be met.
       \EndFor
	    \State Output:  $\{z_{s,c}^{f,t}\}$.
	\end{algorithmic}
\end{algorithm}

\subsection{Complexity analysis}\label{sec:Complexity Alalysis}
\subsubsection{Traditional methods}
To simplify analysis, we ignore the influence of constraint terms in this part.
The variable number $V$ of problem $\boldsymbol{P_1}$, $\boldsymbol{P_2}$, $\boldsymbol{P_3}$ and $\boldsymbol{P_4}$ are $SC+SCBT+SCT$, $SC$, $SCBT$ and $SCT$, respectively.
Referring to \cite{mosek_analysis}, the complexity of traditional methods, i.e., solvers or branch and bound methods, exponentially increases with numbers of variable size, which can be expressed as $\mathcal{O}(2^V)$.

\subsubsection{Our proposal}
In inter-satellite handover decision algorithm (Algorithm~\ref{alg:handover_Algorith}), the maximum number of the first and the second swap operations are $C(C-1)$ and $CS$, respectively.
Due to $S<C$, the complexity of Algorithm~\ref{alg:handover_Algorith} is $\mathcal{O}(N^\prime C^2)$, where $N^\prime$ is the maximum iteration number.
The complexity of beam hopping design algorithm (Algorithm~\ref{alg:msia_Algorithm}) mainly relies on checking the accessible state of vertices and their adjacent vertices in all slots.
The maximum number of vertices and their adjacent vertices are $BC$ and $BC(BC-1)/2$, respectively.
Hence, the complexity of Algorithm~\ref{alg:msia_Algorithm} is $\mathcal{O}(TB^2C^2)$.
For satellite-terrestrial spectrum sharing algorithm (Algorithm~\ref{alg:spectrum_sharing_Algorith}), the maximum number of iterations and the population size of sparrows are  $N^{\prime \prime}$ and $N_{pop}$, respectively.
In each iteration, the numbers of the sparrow update and fitness calculation operations are less than $2.5N_{pop}$.
In addition, the number of adaptive crossover operations is $N^{\prime \prime \prime}$, and $N^{\prime \prime \prime}< N_{pop}$.
Moreover, the last greedy search step requires search all $z_{s,c}^{f,t}$ and corresponding number is $SB$.
Hence, the complexity of Algorithm~\ref{alg:spectrum_sharing_Algorith} is $\mathcal{O}(N^{\prime \prime}N_{pop}+SB)$.
Based on the above analysis, the complexity of the proposed beam management approach is $\mathcal{O}(N^\prime C^2+TB^2C^2 +N^{\prime \prime}N_{pop})$.

\section{Simulation Results and Analysis}\label{sec:sim}

\subsection{Simulation parameters and performance metrics}

To evaluate the performances of the proposed beam management approach, we conduct massive simulations in LEO satellite network scenarios with 1200 LEO satellites, 20 earth-fixed beam cells, and 200 terrestrial clusters.
Satellites are evenly distributed in 30 orbit planes, whose inclination and altitude are $53^{\circ}$ and 550 km, respectively.
Moreover, beam cells are generated in a rectangular area and arranged in four rows and five columns, where each cell is hexagonal, and the inter-cell distance is 34.6 km.
The first beam cell locates at $(30^{\circ}E,20^{\circ}N)$.
Terrestrial base station clusters are randomly generated in circles with radius 50 km, whose center points locate at the centers of beam cells.
Moreover, inter-cluster distance is more than 10 km.
Transmission load $l_{j}^f$ of base station clusters are randomly generated in $[0.4, 0.6]$.

LEO satellites equip with four beams and operate at 20 GHz.
The bandwidth of LEO satellites and terrestrial networks are 200 MHz and 80 MHz, respectively.
Moreover, two orthogonal antenna polarization configurations are considered.
Terrestrial users are handheld terminals with peak gain 0 dB, and the antenna models of LEO satellites and satellite users refer to \cite{38.811} and \cite{ue_Antenna}, respectively.
The channel model adopts the model in 3GPP 38.811 \cite{38.811}.
In addition, the INR threshold of inter-beam interference $I_{s}^{th}$ and satellite-terrestrial interference $I_{g}^{th}$ are preset to -5 dB and -10 dB, respectively.
Here, the minimum elevation angle between LEO satellites and satellite users is $35^{\circ}$.
To simplify simulations, all beam cells are consistently served by two satellites, and target SNR of all beam cells is set to 12 dB.
There are 200 time slots per epoch with a duration of 200 ms, and simulations last 66.67 minutes for evaluating long-term beam management performances.
The tolerable maximum inter-satellite handover frequency threshold $\overline{H}$ is set to 0.004.
We set resource utilization rate threshold $\sigma_0$ and imbalance index $\tau_0$ to 0.9 and 2, respectively.
Main parameters and the normalized average demand among beam cells are summarized in Table~\ref{tab:1} and \ref{tab:2}, respectively.

\begin{table}[htbp]
\centering
\caption{Simulation parameter setting}
\label{tab:1}
\footnotesize
\begin{tabular}{llllll} \toprule
\textbf{Parameter}                      & \textbf{Value} \\
\midrule
The number of satellite orbits            & 30  \\
The number of LEO satellites in an orbit      & 40  \\
The number beams per satellite            & 4   \\
The number of beam cells  & 20  \\
The number of base station clusters  & 200  \\
The range of transmission load  $l_{j,f}$     & [0.4, 0.6] \\
Antenna polarization configuration number              &2    \\
Orbit altitude                            & 550 km \\
Orbit inclination               & $53^{\circ}$  \\
The location of the first beam cell & $(30^{\circ}E,20^{\circ}N)$\\
$I_{s}^{th}$                          & -5 dB   \\
$I_{g}^{th}$                          & -10 dB  \\
$SNR_c^f$                               & 12 dB   \\
$\overline{H}$                          & 0.004 \\
The radius of a beam cell        & 34.6 km   \\
Bandwidth  of LEO satellite networks    & 200 MHz \\
Bandwidth of terrestrial networks     & 80 MHz  \\
Center operation frequency            & 20 GHz   \\
Minimum elevation angle             & $35^{\circ}$\\
The number of slots in a epoch & 200    \\
The epoch number           &20000   \\
Resource utilization rate threshold $\sigma_0$ &0.9 \\
Imbalance index $\tau_0$ &2\\
\bottomrule
\end{tabular}
\end{table}

\begin{table}[htbp]
\centering
\caption{Normalized transmission demand among beam cells}
\label{tab:2}
\scriptsize
\begin{tabular}{lllllllllll}  \toprule
\textbf{ Beam cell}  & 1 & 2 & 3 & 4 & 5 & 6 & 7\\
\midrule
\textbf{Value}$\times 10^{-2}$ & 2.21 & 6.36 & 6.36 & 3.25 & 7.40 & 3.25 & 2.09 \\
\midrule
\textbf{ Beam cell}   & 8 & 9 & 10 & 11 & 12 & 13 & 14\\
\midrule
\textbf{Value}$\times 10^{-2}$  & 4.28 & 4.29 & 5.32 & 8.44 &  4.28 & 3.25 & 2.21 \\
\midrule
\textbf{ Beam cell}   & 15 & 16 & 17 & 18 & 19 & 20 & -\\
\midrule
\textbf{Value}$\times 10^{-2}$ & 6.36 & 7.40 & 3.12 & 8.44 & 4.28 & 7.40 &-\\
\bottomrule
\end{tabular}
\end{table}

The performance metrics include queue length, inter-satellite handover frequency, and objective value of problem $\boldsymbol{P_0}$.
Moreover, to succinctly characterize the performance algorithms over time in figures, we take the average for queue length and handover frequency metrics among all cells.

\begin{figure*}[htbp]
    \centering
        \subfigure[]
		{
                \includegraphics[width=2.25in]{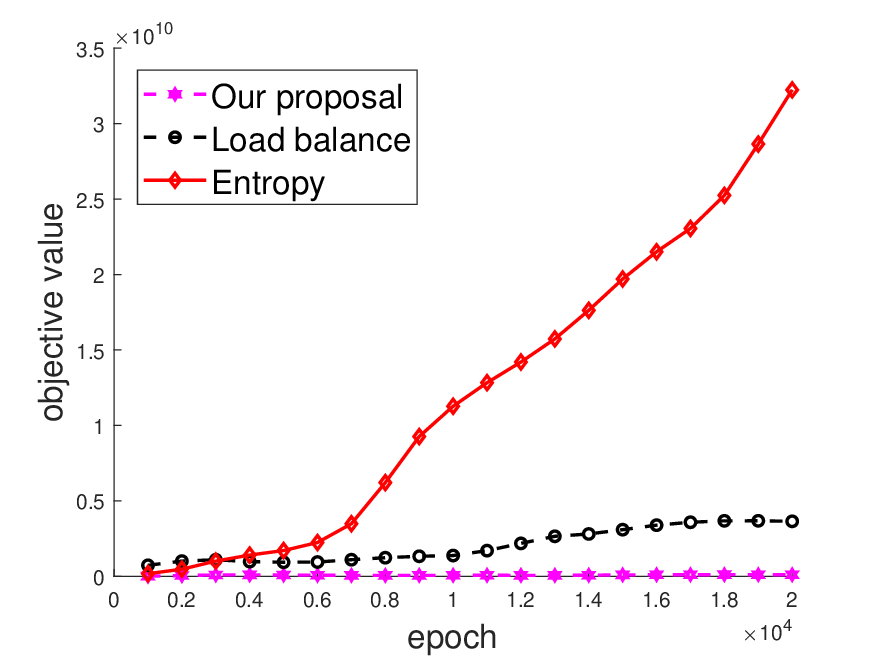}
        }
        \subfigure[]
		{
                \includegraphics[width=2.25in]{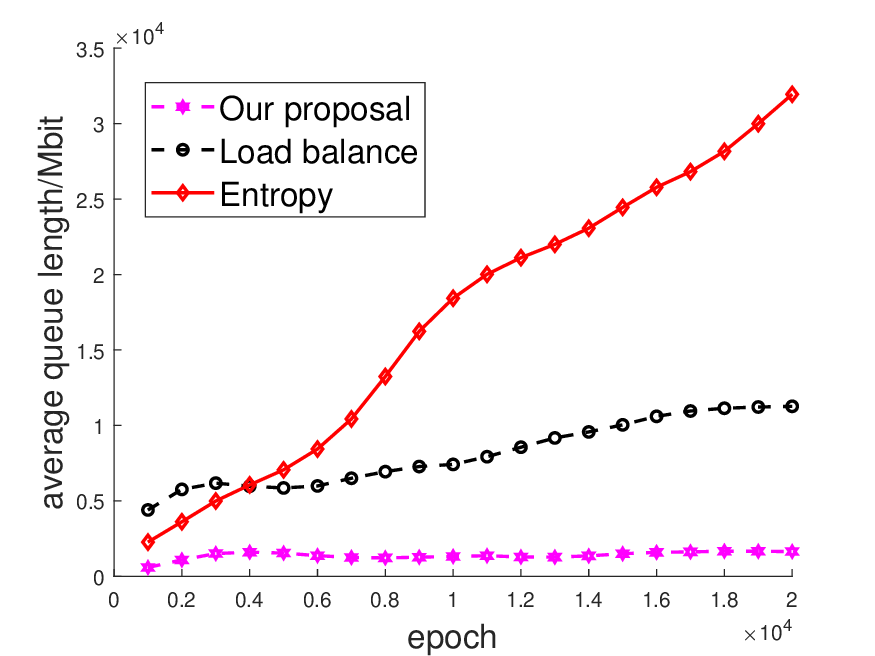}
        }
        \subfigure[]
		{
                \includegraphics[width=2.25in]{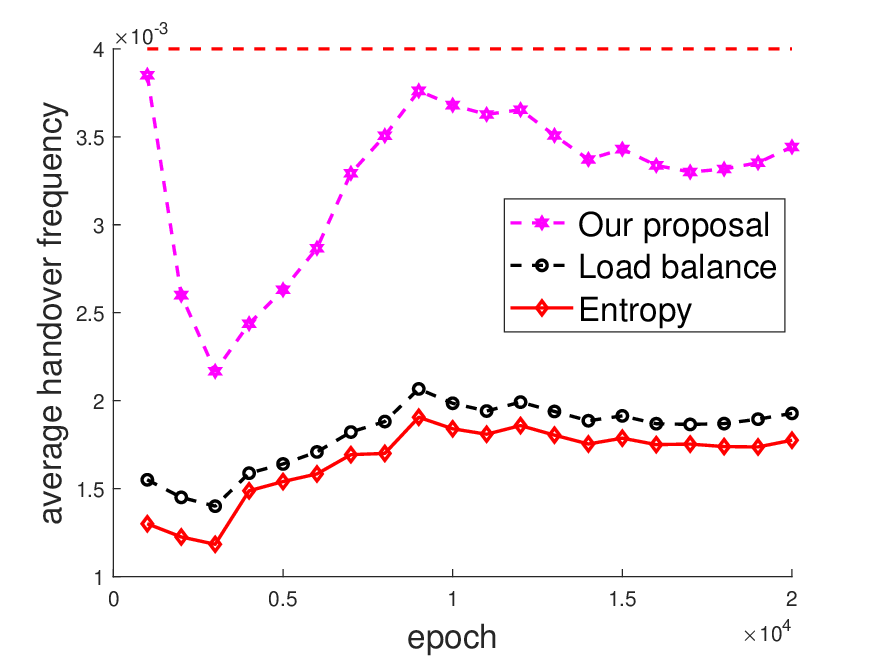}
        }
    \caption{The figures are the results of the performance of inter-satellite handover decision algorithms, where (a) presents the objective value of problem $\boldsymbol{P}_o$, (b) is the result of average data queue length of beam cells, and (c) shows the average inter-satellite handover frequency of all beam cells.}\label{fig5}
\end{figure*}

\begin{figure*}[htbp]
    \centering
        \subfigure[]
		{
                \includegraphics[width=2.25in]{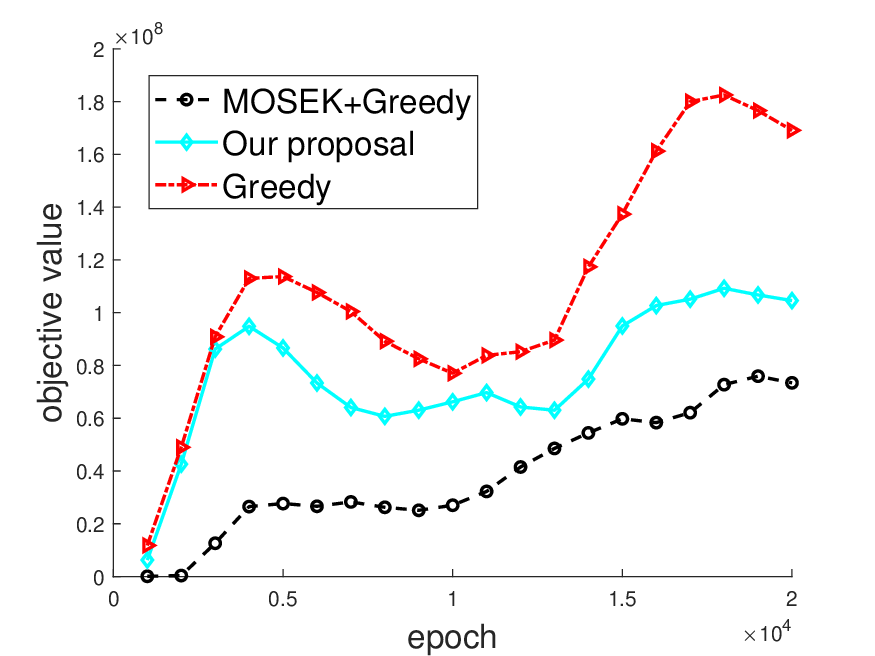}
        }
        \subfigure[]
		{
                \includegraphics[width=2.25in]{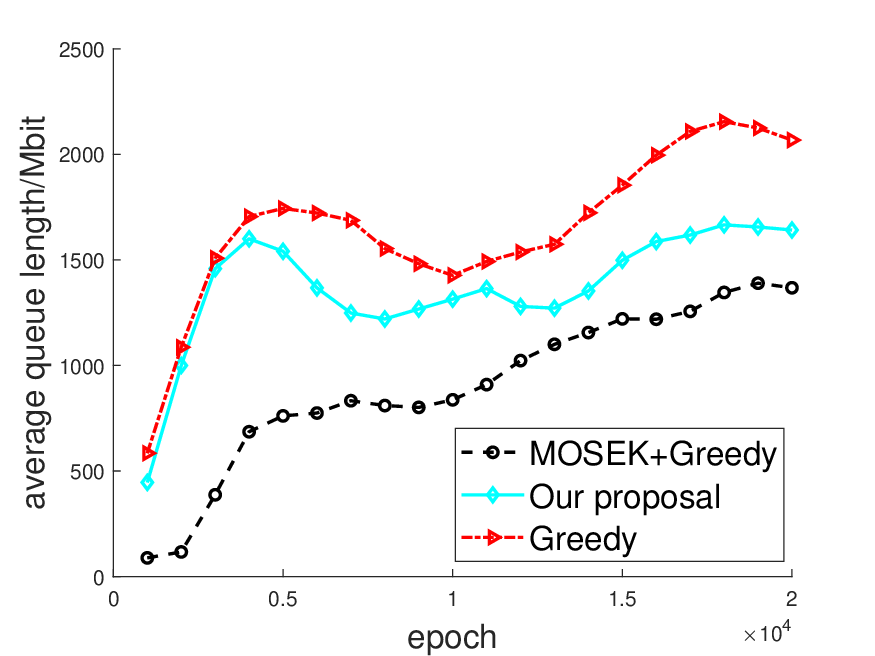}
        }
        \subfigure[]
		{
                \includegraphics[width=2.25in]{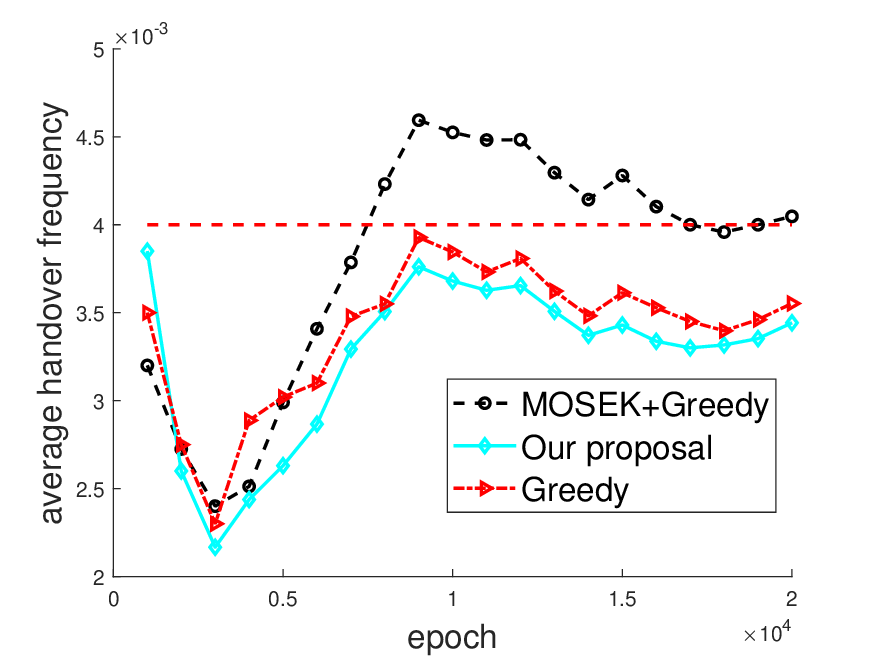}
        }
    \caption{The figures are the results of the performance of beam hopping design algorithms, where (a) presents the original objective value of problem $\boldsymbol{P}_o$, (b) is the result of average data queue length of beam cells, and (c) shows the average inter-satellite handover frequency of all beam cells.}\label{fig6}
\end{figure*}

\begin{figure*}[htbp]
    \centering
        \subfigure[]
		{
                \includegraphics[width=2.25in]{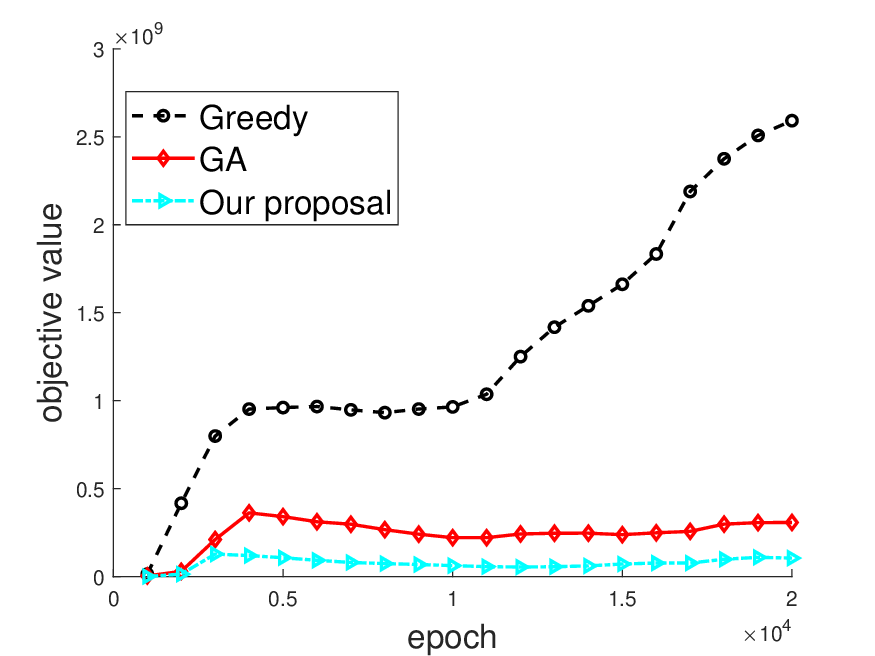}
        }
        \subfigure[]
		{
                \includegraphics[width=2.25in]{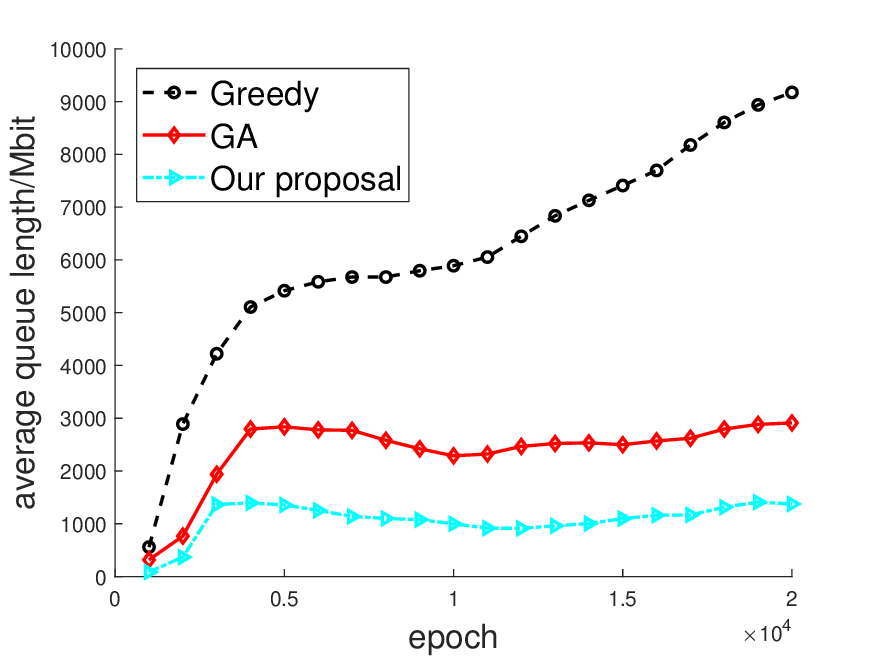}
        }
        \subfigure[]
		{
                \includegraphics[width=2.25in]{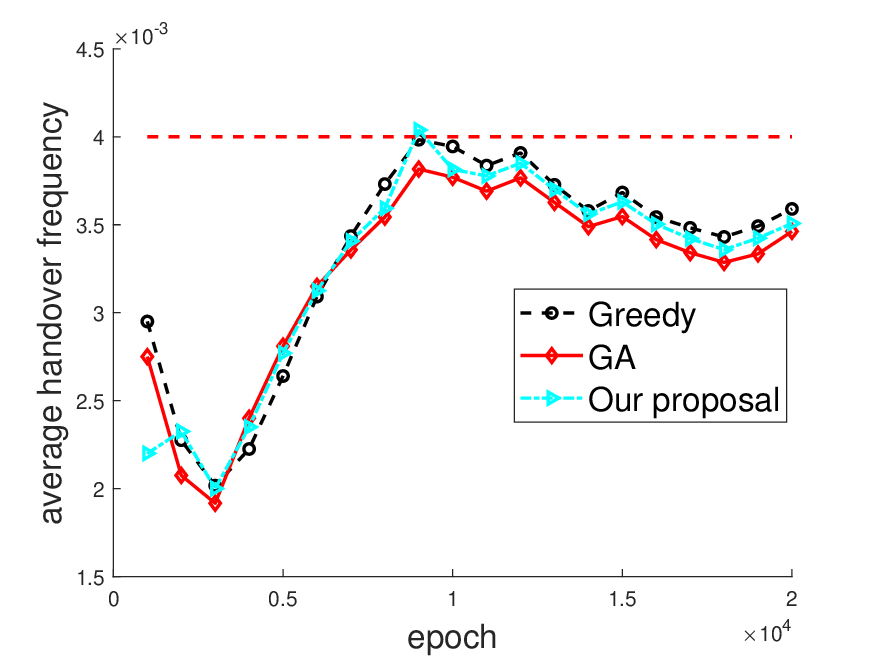}
        }
    \caption{The figures are the results of the performance of satellite-terrestrial spectrum sharing algorithms, where (a) presents the objective value of problem $\boldsymbol{P}_o$, (b) is the result of average data queue length of beam cells, and (c) shows the average inter-satellite handover frequency of all beam cells.}\label{fig7}
\end{figure*}

\begin{figure*}[htbp]
    \centering
        \subfigure[]
		{
                \includegraphics[width=2.25in]{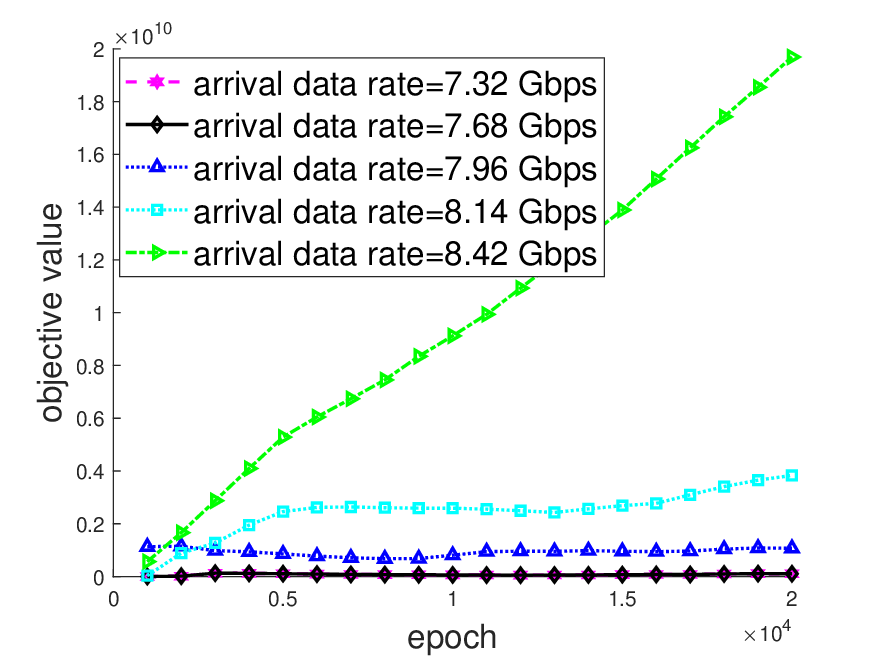}
        }
        \subfigure[]
		{
                \includegraphics[width=2.25in]{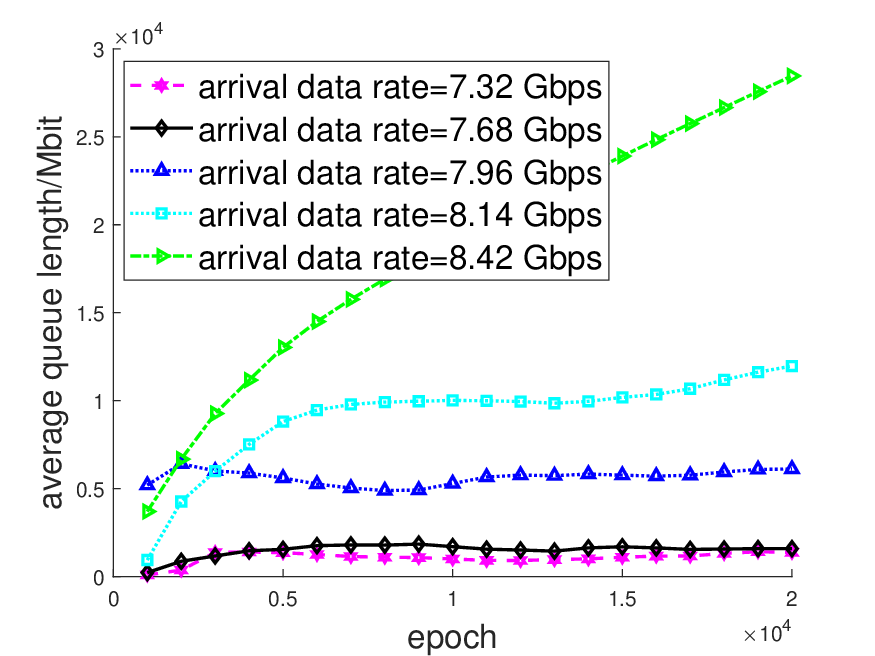}
        }
        \subfigure[]
		{
                \includegraphics[width=2.25in]{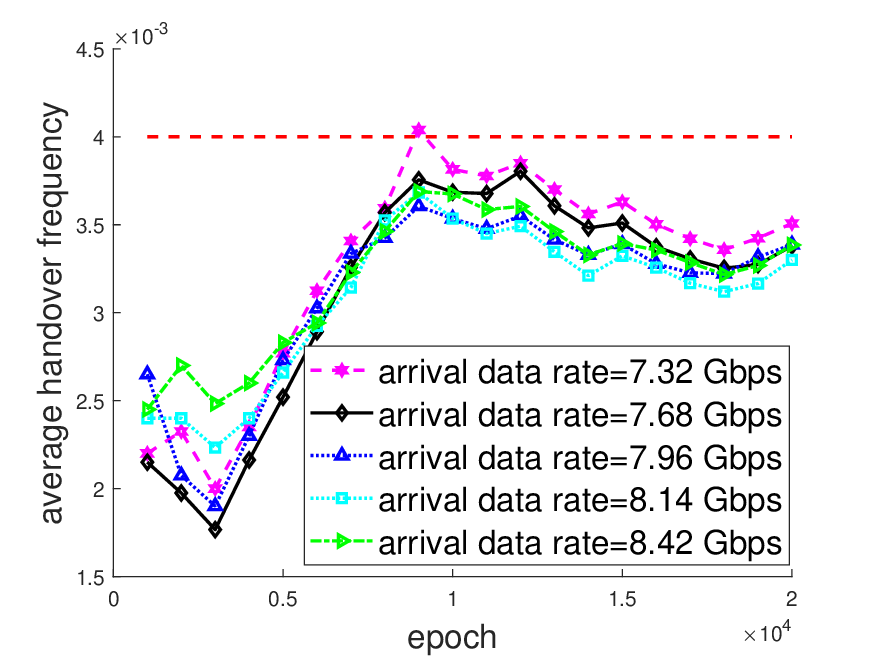}
        }
    \caption{The performance of our proposal under different the average sum of arrival data rate, where (a) presents the original objective value of problem $\boldsymbol{P}_o$, (b) is the result of average data queue length of beam cells, and (c) shows the average  inter-satellite handover frequency of all beam cells.}\label{fig8}
\end{figure*}

\subsection{The performance comparisons with existing methods}\label{sec:handover}

In this part, we first evaluate the performance of the proposed inter-satellite handover decision algorithm, where two benchmark schemes are adopted as follows.
\begin{itemize}
	\item \emph{Load Balance:} The visible satellite with the minimum traffic load is allocated to beam cells~\cite{multi_sat2}.
	\item \emph{Entropy:} Inter-satellite handover procedure is implemented by a multi-attribute handover decision scheme based on entropy in~\cite{multi-attribute}.
\end{itemize}
Specifically, inter-satellite handover procedures of benchmark schemes are only triggered when network topology changes.
For a fair comparison, benchmark schemes and our proposal adopt the proposed beam hopping deign algorithm to formulate beam hopping design, and satellite-terrestrial spectrum sharing operations are ignored in this stage.
The average sum of arrival data rate of all beam cells is 6.52 Gbps and $V=100$.
Since LEO satellites can serve beam cells within a long duration, random traffic arrival of cells leads to load imbalance among satellites.
A short data queue length level means that LEO satellites have the capacity to transmit more data under the given spectrum resource.
Fig.~\ref{fig5} (a) and (b) illustrate that our proposal surpasses benchmark schemes.
Specifically, the objective value of problem $\boldsymbol{P}_o$ and average beam cell data queue length of our proposal are reduced 97.13$\%$ and 84.47$\%$ compared with load balance strategy, which demonstrates that LEO satellite networks proactively trigger inter-satellite handover procedure contributing to high spectral efficiency.
As shown in Fig.~\ref{fig5} (c), our proposal obtains a better performance at the cost of handover frequency.
Remarkably, our proposal satisfies inter-satellite handover frequency constraint.

Then, we evaluate the performance of the proposed beam hopping design algorithm, and simulation parameters are consistent with the evaluation in Fig.~\ref{fig5}.
Here, two baselines are considered as follows.
\begin{itemize}
	\item \emph{Greedy scheme:} The gateway station greedily allocates serving beams to cells with large data queue~\cite{greedy1}.
	\item \emph{MOSEK+Greedy:} The beam hopping design is obtained based on MOSEK solver and greedy search step in\cite{multi_sat2}.
\end{itemize}
In Fig.~\ref{fig6} (a) and (b), it can be seen that the beam hopping design algorithm based on MOSEK solver and greedy search achieves the highest spectral efficiency.
However, computational complexity is a core factor for practical applications.
Due to exponential complexity of MOSEK solver~\cite{Yuan,mosek_analysis}, it is hard to be deployed in practical LEO satellite networks.
Since LEO satellites have excellent spatial isolation in the first 3000 epochs, the performance curves of our proposal and greedy algorithm are tightly close to each other.
With interference situations exacerbating, our proposal can reduce the average data queue length and the objective value of problem $\boldsymbol{P_0}$ by 20.6$\%$ and 38.16$\%$ compared with greedy algorithm at the 20000-th epoch, respectively.
Moreover, as shown in Fig.~\ref{fig6} (c), the beam hopping design outputted by the baseline based on the MOSEK solver and greedy search violates inter-satellite handover frequency constraint, whereas our proposal satisfies this constraint and achieves the minimum handover frequency.
Hence, our proposal with low complexity is more suitable for practical LEO satellite networks.

Later on, we present the performance evaluation result of the satellite-terrestrial spectrum sharing algorithm, where the average sum of arrival data rate of all beam cells is 7.32 Gbps and $V=100$, and two baselines are considered.
\begin{itemize}
	\item \emph{Greedy scheme:} The gateway station greedily sets $z_{s,c}^{f,t}=1$ for beam cells with large data queue lengths~\cite{greedy1}.
	\item \emph{GA:} Problem $\boldsymbol{P}_4$ is solved by genetic algorithm~\cite{ga1}.
\end{itemize}
Fig.~\ref{fig7} (a) illustrates the performance comparison result of the concerned objective value of problem $\boldsymbol{P}_o$.
It can be seen that our proposal outperforms baselines and obtains a high service satisfaction.
As shown in Fig.~\ref{fig7} (b), our proposal reduces the average data queue length by over 50$\%$ compared with genetic algorithm and greedy scheme.
Fig.~\ref{fig7} (c) presents the result of average inter-satellite handover frequency.
Specifically, our proposal satisfies the inter-satellite handover frequency constraint and only requires an additional handover operation compared with GA scheme.
Hence, the proposed satellite-terrestrial spectrum sharing algorithm can significantly improve spectral efficiency.

\subsection{Performance validation of the proposed method }

Subsequently, we evaluate the performance of the proposed beam management approach under different arrival data rates, and set $V=100$.
In Fig.~\ref{fig8} (a) and (b), we can reasonably conclude that our proposal can stabilize the data queue of beam cells when the arrival date rate is less than 7.96 Gbps.
Furthermore, the maximum network capacity in a scenario without interference is 8.08 Gbps , which can be calculated by equation (\ref{eq:rate}).
Hence, simulation results prove that our proposal can achieve high spectral efficiency with low computational complexity.
In addition, with the network load increasing, the average inter-satellite handover frequency slightly declines.
This is because a heavy network load usually leads to a high resource utilization rate, which lowers the invoking frequency of the proposed inter-satellite handover decision algorithm.

\begin{figure}[htbp]
    \centering

        \subfigure[]
		{
                \includegraphics[width=1.62in]{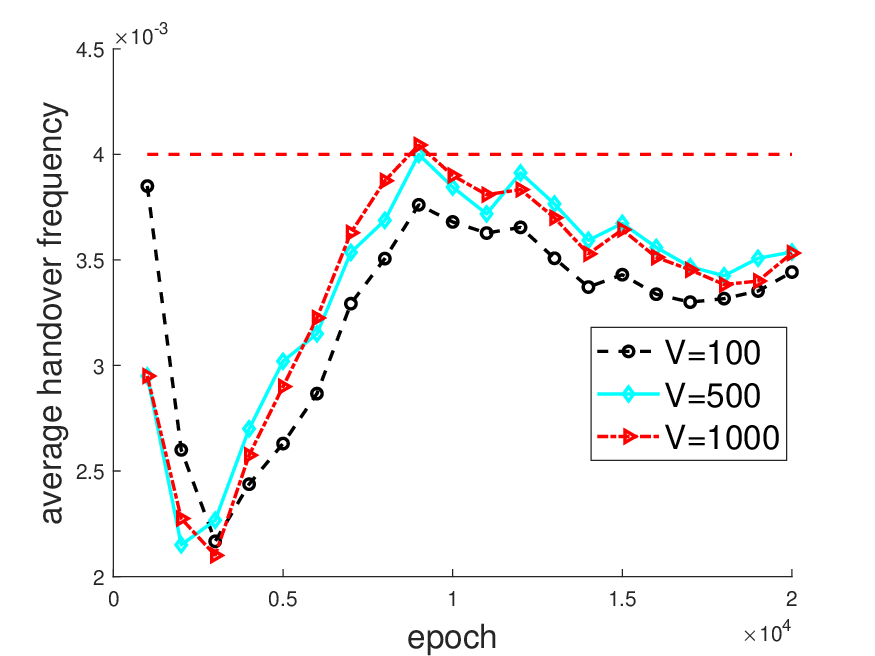}
        }
        \subfigure[]
		{
                \includegraphics[width=1.62in]{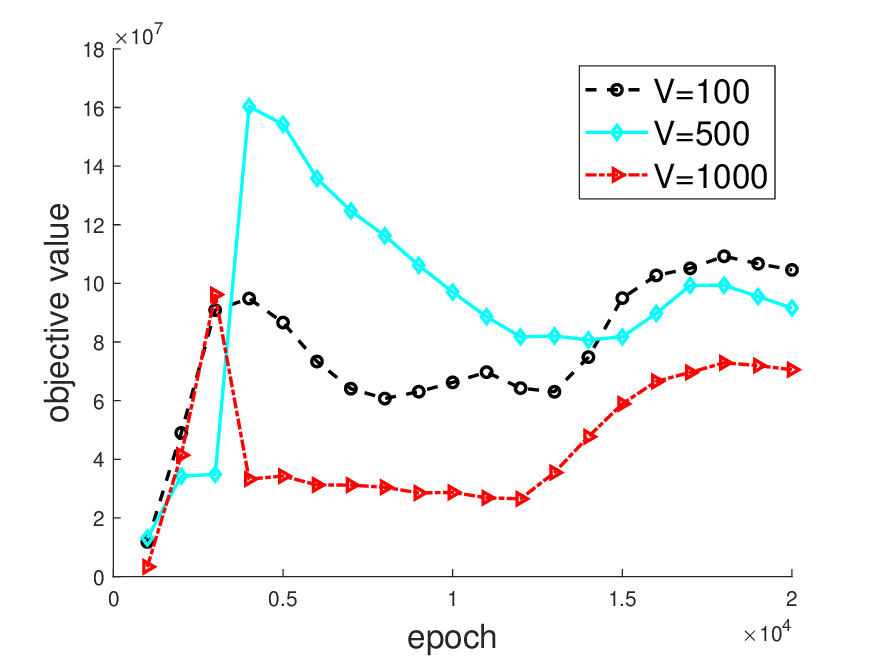}
        }

    \caption{The performance of our proposal under different $V$, where (a) shows the average inter-satellite handover frequency of all beam cells, and (b) presents the objective value of problem $\boldsymbol{P}_o$.}\label{fig9}
\end{figure}

Finally, we investigate the influence of parameter $V$ on inter-satellite handover frequency and objective value of problem $\boldsymbol{P}_o$.
Moreover, we ignore satellite-terrestrial spectrum sharing operations and set the average sum of arrival data rate of all beam cells to 6.52 Gbps.
According to Fig.~\ref{fig9}, it can be observed that a larger $V$ corresponds to a higher inter-satellite handover frequency.
The reason for this phenomenon is that a large $V$ forces the LEO satellite networks to pay more attention to control load distributions among satellites, whereas the effect of handover frequency constraint becomes weakened.

\section{Conclusions}\label{sec:con}
For a high service satisfaction with low inter-satellite handover frequency, this article has studied the long-term beam management approach for LEO satellite networks with dynamic topology, time-varying loads and complex inter-beam/satellite-terrestrial interference.
To overcome the issue caused by complex interference situations, we first proposed interference mitigation strategies to simplify interference analysis and support spectrum sharing.
Since the concerned beam management problem considers long-term service satisfaction and inter-satellite handover frequency constraints, we adopted Lyapunov drift to transform the primal problem into a series of per-epoch problems.
Considering that the transformed problem was NP-hard, we have further decomposed it into three subproblems, including inter-satellite handover decision problem, beam hopping design problem, and satellite-terrestrial spectrum sharing problem.
Subsequently, conditional handover triggering mechanism and inter-satellite handover algorithm have been designed to control the handover frequency and load distributions among LEO satellites.
Based on the proposed interference mitigation strategies, we have designed low-complexity beam hopping design and satellite-terrestrial spectrum sharing algorithms, contributing to high service satisfaction and spectral efficiency.
Finally, the simulation results have verified that the proposed beam management approach achieved a better average data queue length and service satisfaction than benchmark schemes.
Meanwhile, our proposal satisfied the maximum inter-satellite handover frequency constraints.


\end{document}